\documentclass[conference]{IEEEtran}
\usepackage{times}

\usepackage[numbers]{natbib}
\usepackage{multicol}

\usepackage{amsmath}
\usepackage{amssymb}
\usepackage{amsthm} 
\usepackage{bm}
\usepackage{graphicx}
\usepackage{float}
\let\labelindent\relax
\usepackage{enumitem}
\usepackage[mathscr]{euscript}
\usepackage[version=4]{mhchem}
\usepackage{siunitx}
\usepackage{longtable,tabularx}
\usepackage[footnotesize]{caption}
\usepackage{subcaption}
\usepackage[utf8]{inputenc}
\usepackage{tikz}

\usepackage[colorlinks,urlcolor = black]{hyperref}

\def\myMSFigureScale{0.25}

\def\myLineScale{1}

\renewcommand{\baselinestretch}{0.99}
\usepackage{flushend}

\usepackage{booktabs}
\usepackage{multirow}

\usepackage{setspace}
\usepackage[normalem]{ulem}

\usepackage{thmtools}
\usepackage{thm-restate}

\usepackage{algorithm}
\usepackage[noend]{algpseudocode}
\makeatletter
\def\BState{\State\hskip-\ALG@thistlm}
\makeatother
\algnewcommand{\Or}{\textbf{or}\,}
\algnewcommand\algorithmicswitch{\textbf{switch}}
\algnewcommand\algorithmiccase{\textbf{case}}
\algdef{SE}[SWITCH]{Switch}{EndSwitch}[1]{\algorithmicswitch\ #1\ \algorithmicdo}{\algorithmicend\ \algorithmicswitch}%
\algdef{SE}[CASE]{Case}{EndCase}[1]{\algorithmiccase\ #1}{\algorithmicend\ \algorithmiccase}%
\algdef{SE}[ELSE]{Else}{EndElse}[1]{\algorithmicelse\ #1}{\algorithmicend\ \algorithmicelse}%
\algtext*{EndSwitch}%
\algtext*{EndCase}%
\algtext*{EndElse}%

\newcommand{\set}[1]{\{#1\}}

\makeatletter
\newcommand*{\defeq}{\mathrel{\rlap{\raisebox{0.3ex}{$\m@th\cdot$}}\raisebox{-0.3ex}{$\m@th\cdot$}}=}
\makeatother

\begin{document}


\title{
 CBS-Budget (CBSB): A Complete and  Bounded Suboptimal Search
 for Multi-Agent Path Finding 
}
\author{\authorblockN{Jaein Lim}
	\authorblockA{School of Aerospace Engineering\\
		Georgia Institute of Technology\\
		Atlanta, Georgia 30332--0150\\
		Email: jaeinlim126@gatech.edu}
	\and
	\authorblockN{Panagiotis Tsiotras}
	\authorblockA{School of Aerospace Engineering\\
		Institute for Robotics \& Intelligent Machines\\
		Georgia Institute of Technology\\
		Atlanta, Georgia 30332--0150\\
		Email: tsiotras@gatech.edu}}

%
%

\maketitle

\begin{abstract}
Multi-Agent Path Finding (MAPF) is the problem of finding a collection of collision-free paths for a team of multiple agents while minimizing some global cost, such as the sum of the time travelled by all agents, or the time travelled by the last agent. 
Conflict Based Search (CBS) is a leading complete and optimal MAPF solver which lazily explores the joint agent state space, using an admissible heuristic joint plan. 
Such an admissible heuristic joint plan is computed by combining individual shortest paths found without considering inter-agent conflicts, and which becomes gradually more informed as constraints are added to individual agents' path planning problems to avoid discovered conflicts. 
In this paper, we seek to speedup CBS by finding a more informed heuristic joint plan which is bounded from above.
We first propose the budgeted Class-Ordered A* (bCOA*), a novel algorithm that finds the shortest path with minimal number of conflicts that is upper bounded in terms of length. 
Then, we propose a novel bounded-cost variant of CBS, called CBS-Budget (CBSB) by using a bCOA* search at the low-level search of the CBS and by using a modified focal search at the high-level search of the CBS.
We prove that CBSB is complete and bounded-suboptimal.
In our numerical experiments, CBSB finds a near optimal solution for hundreds of agents within a fraction of a second.  
CBSB shows state-of-the-art performance, 
comparable to Explicit Estimation CBS (EECBS), an enhanced recent version of CBS.
On the other hand,
CBSB is easier to implement than EECBS, since only two priority queues at the high-level search are needed as in Enhanced CBS (ECBS). 

\end{abstract}
\IEEEpeerreviewmaketitle

%
%
\section{Introduction}

Multi-Agent Path Finding (MAPF) is the problem of finding collision-free paths for a team of mobile agents 
from their individual start locations to their individual goal locations, 
while minimizing the sum of their travel times or the time travelled by the last agent. 
Inspired from numerous real-world applications such as warehouse logistics~\cite{Wurman2008}, automated valet parking~\cite{Okoso2019}, video games~\cite{Silver2005}, robotics~\cite{bennewitz2002} and traffic management~\cite{dresner2008,Ho2019}, MAPF has drawn significant attention recently from various research communities.
Solving MAPF with theoretical guarantees is imperative for operating many agents in the same cluttered environment, where uncoordinated plans or prioritized plans can result in deadlocks.

One approach to find a collision-free path minimizing some global cost is to search in the joint agent search space, that is, in the Cartesian product of the individual agents’ search spaces~\cite{ryan2008,standley2010}. 
However, the branch factor of the joint search space grows exponentially with the number of agents, prohibiting the use of a regular heuristic search algorithm like A*. 
For example, for a grid world with 30 agents, where each agent can either move up, down, left, right, or wait,
the joint state has about $9.31 \times 10^{20}$ neighbor states; even enumerating those neighbors to begin the search becomes intractable~\cite{felner2017}.  

Previous complete and optimal MAPF solvers seek to explore only the neighbors of a joint state that could potentially yield an optimal solution~\cite{Wagner2011, Felner2012, Standley2011, Solovey2016, shome2020}. 
Conflict Based Search (CBS)~\cite{Sharon2015} is a state-of-the-art optimal MAPF solver that utilizes a lazy algorithm to delay generation of irrelevant neighbors. 
CBS delays the generation of irrelevant neighbors by decomposing the joint search space to multiple individual path planning problems.
CBS first computes the shortest path for each individual agent without considering other agents.
The combination of the shortest paths is an admissible heuristic plan of the joint space which does not overestimate the true optimal solution cost.
CBS then checks if the combination of the shortest paths, or the heuristic plan is indeed conflict free.
When the heuristic plan is found to be invalid because of inter-agent conflicts between paths, then CBS generates two 
new heuristic plans by only replanning for the two conflicting agents one at a time, with the additional constraint to avoid conflict at the low-level search\footnote{In case a conflict of three or more agents occurs, CBS picks any two agents to resolve the conflict at a time.
The remaining conflicts are resolved in the subsequent searches.}.
CBS continues to choose the heuristic plan with the minimum cost for examination of inter-agent conflicts until a conflict-free plan is found.
This minimal generation of admissible heuristic plans makes the exploration of the joint state space efficient and complete. 


CBS may still become intractable, however, if there exist many infeasible plans due to conflicts whose cost is yet lower than the optimal solution.
As a best-first search, CBS eliminates all those plans shorter than the optimal solution until the optimal solution is found. 
Recent improvements have made significant progress on CBS by increasing the lower bound of the optimal solution more quickly, that is, by finding a more informed admissible plan using domain-specific knowledge~\cite{Boyarski2015a, Li2020, Zhang2020, Li2019, Boyarski2020b, Felner2018}.
Solving the MAPF problem optimally is, however, NP-hard~\cite{Yu2013}, and the problem still remains unsolvable for a large number of agents with reasonable computational resources. 

If a bounded-suboptimal solution is allowed, then CBS can be modified to solve large problem instances. 
The main idea is to generate more informed, possibly inadmissible, and yet bounded, heuristic joint plans that could result in a better progress during search.
This improvement is the main topic of the paper. 

Enhanced CBS (ECBS)~\cite{Barer2014} is a recent variant of CBS that 
computes a bounded-suboptimal heuristic plan using a focal search. 
Focal search~\cite{Pearl1982} maintains two priority queues, one for bounding the solution cost using an admissible heuristic and the other for finding an informed and possibly inadmissible path using another heuristic. ECBS uses a focal search in both the low and high-level search of CBS to generate a more informed heuristic plan, that is a bounded-suboptimal plan with a fewer number of conflicts.

Recently, a further improvement to ECBS was made by replacing the high-level focal search of ECBS with Explicit Estimation Search (EES)~\cite{Li2021}. 
Focal search becomes inefficient when the two heuristics are negatively correlated. EES~\cite{Thayer2011} mitigates the negative correlation between the two heuristics of focal search, using a third heuristic which may be potentially inadmissible.
Explicit Estimation CBS (EECBS)~\cite{Li2021} uses EES in the high-level search of CBS to avoid the high-level focal search being stuck locally by considering the potential cost increment upon resolving a conflict.  
Maintaining three different priority queues requires more memory usage and adds complexity to the algorithm. 
In fact, the number of frontier nodes (i.e., the number of joint plans) is exponential in the depth of the CBS search (i.e., the number of resolved conflicts)~\cite{Boyarski2020}, which makes memory limit 
more of a bottleneck than the time limit in CBS implementations.
In addition, incorporating modern implementations of CBS to EECBS is not trivial, as specific rules need to be carefully considered,
depending on which priority queue a heuristic joint plan is chosen, 
adding more complexity to the algorithm. 

In this paper, we propose a novel and simple bounded-cost variant of CBS, called CBS-Budget (CBSB). 
CBSB uses a budgeted Class-Ordered A* (bCOA*) search on the low-level of CBS to produce a more informed and still bounded heuristic plan.
Similarly to focal search, bCOA* finds the shortest path for each agent with minimal inclusion of conflicts with other agents and whose cost is upper bounded either by a given budget or by the shortest path length.
The budget for each agent is incrementally updated after each search if no shorter path exists, making the budget-induced suboptimality bound more informative.
Unlike focal search, bCOA* only uses a single priority queue by leveraging the ideas of colored planning~\cite{Wooden2006, Lim2021a, Lim2021b},
providing a tighter lower bound than focal search, since bCOA* guarantees to return the shortest path if no path shorter than a given budget exists.
This property of bCOA* allows the high-level search of CBSB to explore more informative heuristic plans according to the number of conflicts, leading to a faster computation time to find a bounded suboptimal solution.
We prove the theoretical properties of bCOA* and the completeness and bounded-suboptimality of CBSB. 
In our experiments on standard MAPF benchmarks~\cite{Stern2019}, CBSB shows state-of-the-art performance compared to the modern implementations of complete and bounded-(sub)optimal MAPF solvers such as CBS, ECBS, and EECBS.

\section{Preliminaries and Relation to Prior Work}
In this section, we formalize the definition of MAPF and provide some background on the baseline CBS algorithm and its improved variants. 

\subsection{Multi-Agent Path Finding (MAPF)}

MAPF is the problem of finding a set of collision-free paths for multiple agents in the same graph. Formally, MAPF is defined by a graph $G=$($V,E$) and a set of $m$ agents $\set{a_1,…,a_m}$, where each agent $a_i$ has a start vertex $s_i\in V$ and a target vertex $t_i\in V.$ 
Time is discretized, and at each timestep, every agent can either move to an adjacent vertex or wait at the current vertex. 
Without loss of generality, we assume that all actions have uniform cost. 
The agents remain at their target vertices after completion.

A path $p_i$ for agent $a_i$ is a sequence of adjacent vertices from $s_i$ to $t_i$. The cost of a path $p_i$ is its length. 
We say that two agents have a conflict (or equivalently, a collision) if they occupy the same vertex at the same timestep or if they move along the same edge at the same timestep. 
A plan is a set of individual agent paths $\set{p_1,…,p_m}$, and a MAPF solution is a plan that is conflict free. 
The cost of a plan is the sum of its paths lengths. 
An optimal solution is a solution whose cost is minimal.\footnote{Our algorithm is not limited to this \textit{flowtime} formulation. The cost of a plan can be defined as the longest path length to find the optimal \textit{makespan} plan.}

\subsection{Conflict Based Search (CBS)}

Conflict Based Search (CBS)~\cite{Sharon2015} is a two-level MAPF solver. 
At the low level, each agent finds the shortest path on a time-expanded graph (TEG)~\cite{surynek2017} with some variant of A*.
Each vertex in TEG is a pair of location and time, and each edge in TEG represents a move from a vertex at one timestep to another vertex at the next timestep or a wait at a vertex for one timestep. 
Each agent may be subject to a set of constraints that prohibit this agent from occupying a vertex or from moving through an edge at a certain timestep.
The constraints are generated at the high-level to avoid a conflict between the two conflicting agents. 
Specifically, at the high level, CBS performs a best-first search by constructing a binary constraint tree (CT), 
where each CT node $N$ contains a plan (possibly with conflicts) consisting of the individual paths computed at the low level, along with the corresponding sum of the costs, denoted as $cost(N)$. 
When CBS chooses a CT node with the minimum $cost$ to expand, CBS checks if a conflict occurs in the plan of that node. 
If a conflict exists, then CBS imposes constraints on the two conflicting agents one at a time in each of the two children of the CT node to avoid the conflict. 
With this additional constraint added, the low-level CBS search finds the next shortest path satisfying the constraint for the agent in the child CT node. 
CBS chooses the next best node prioritized with the best possible cost, until no conflict is found.  

In essence, CBS generates an admissible heuristic plan by finding the shortest path of individual agents, and if the plan is not conflict free, then CBS generates the next best plan by replanning the individual shortest paths of the conflicting agents with an additional constraint.
The two-level search of CBS produces an admissible heuristic plan that becomes more informed as conflicts are discovered.
Each CT node contains a heuristic plan, and a child CT node contains a more informed heuristic plan with additional constraints.
The generated plans are prioritized according to their $f=cost+h$ values and the CT node with the lowest $f$ is chosen for expansion using an A*-like search, where $cost$ is the current plan cost and $h$ is an admissible heuristic cost-increment to resolve the conflicts in the CT node. 
Hence, the chosen CT node for expansion always contains an admissible heuristic plan, ensuring the optimality of CBS.

\subsection{Modern Improvements}

\textit{Prioritizing Conflicts}~\cite{Boyarski2015a}:
The $cost$-value of a CT node monotonically increases
along the depth of CT, because each individual agent finds a path with monotonically increasing length with an additional constraint at the low level search.
If one can find a conflict that would result in a large $cost$-value increase, then CBS can make good progress toward the goal CT node without sacrificing optimality. 
This observation has motivated prioritization of conflicts so as to select a conflict that would maximally increase the $cost$-value of a child CT node. 
The potential cost increase of resolving a conflict can be determined by building a Multi-Valued Decision Diagram (MDD)~\cite{Sharon2013} for each agent, an acyclic directed graph that consists of all shortest paths of the agent.

\textit{Symmetry Reasoning}~\cite{Li2020}:
Imposing a more informed constraint to resolve a conflict can benefit progress of CBS. 
Symmetry reasoning is a technique that identifies repeated collisions between two agents due to symmetric paths. 
Domain-specific knowledge can be used to generate symmetry-breaking constraints to generate a more informed and admissible heuristic plan. 

\textit{Mutex Propagation}~\cite{Zhang2020}:
Symmetry reasoning identifies three types of conflicts (rectangle, corridor, target) and generates hand-crafted constraints. 
Mutex propagation is a technique that automates the generation of symmetry-breaking constraints using MDDs. 
Each level of MDD contains all the reachable states at that time that belong to the shortest paths to the goal satisfying the constraints.   
Two nodes at the same level of MDDs that belong to two different agents are mutex, if no bypassing conflict-free paths exist. Propagating mutex along the MDD can generate a set of more informed constraints by identifying the unreachable set of states in the MDDs. 

\textit{Weighted Dependency Graph (WDG) Heuristic}~\cite{Li2019}:
WDG heuristic is an admissible heuristic for the high-level search of CBS. Vertices in WDG represent agents and edges represent two agents that are dependent.
The dependency of two agents is determined by solving a two-agent MAPF instance for the two agents. 
If the sum-of-cost of the subproblem is larger than their current sum-of-cost, then the two agents are dependent. 
The edge weights are assigned the difference between the two sum-of-costs. 
As a result, the edge-weighted minimum vertex cover of the graph better approximates the true distance-to-go of a CT node.

\textit{Bypassing Conflicts}~\cite{Boyarski2015b}:
%
Bypassing conflict (BP) is a conflict-resolution technique that avoids splitting a CT node into two children CT nodes 
by simply replacing the plan of the parent CT node with a new plan found
in the child CT node if the new plan is better.
Specifically, when expanding a CT node $N$ and generating a child CT node $N'$, BP checks if the child CT node $N'$ has the same cost as the parent CT node $N$ and the number of conflicts in $N'$ is less than the number of conflicts in $N.$ 
If so, then BP replaces $N$ with $N'$ and discards all children CT nodes.
Otherwise, $N$ is split as before. 
Bypassing conflicts without splitting to children nodes produces smaller CTs and results in a reduced runtime.

\subsection{Bounded-Suboptimal Variants}

Finding an optimal solution of MAPF remains difficult for instances with large numbers of agents, even with these modern improvements. 
CBS still cannot find a solution for hundreds of agents within reasonable time or memory limits.
If a bounded-suboptimal solution is allowed, then large-scale MAPF instances can be solved. 
Some of the recent variants of CBS can find a bounded-suboptimal solution much faster than CBS are discussed below.

\subsubsection{Enhanced CBS (ECBS)}

ECBS~\cite{Barer2014} replaces the A* search of CBS with focal search, both at the high-level and at the low-level CBS searches in order to compute a more informed and bounded-suboptimal plan. 
Focal search, or A$^*_\epsilon$~\cite{Pearl1982}, is a bounded-suboptimal search algorithm that maintains two priority queues: OPEN and FOCAL. 
OPEN is a regular A* priority queue that sorts nodes according to their admissible cost estimate $f.$
FOCAL contains a subset of OPEN nodes whose $f$ values are within a user-specified bound
from the current best cost estimate, i.e., FOCAL$=\set{n\in$~OPEN~$: f(n)\leq w\, f(best_f)},$ where $best_f$ is the top node in OPEN with minimum $f$-value and $w > 1$ is a user-specified suboptimality factor.
FOCAL sorts the nodes with a different heuristic $d,$ which may be inadmissible, but potentially more informative. 

Focal search always expands a node from FOCAL with the minimum $d$ value in FOCAL,
which removes this node from both FOCAL and OPEN.
The neighbors of the expanded node are inserted to OPEN and to FOCAL if their $f$ values are within the suboptimality bound, i.e., if $f \leq w\,f(best_f)$.
This search propagation continues until a goal node is expanded.
Since the minimum $f$ value in OPEN is always a lower bound on the optimal solution cost, any node in FOCAL does not overestimate the optimal solution cost $f^*$ by more than $w$. 
Hence, when focal search returns a solution, the cost is at most $w\, f^*.$

ECBS utilizes focal search with the same suboptimality factor $w$ at both the low and the high levels of CBS.  
At the low level, each agent uses a focal search with heuristic $d$ that counts the number of collisions with other paths in the current CT node $N.$
Hence, the low level ECBS finds a path satisfying the constraints in $N$ with the minimum number of collisions with other agents, whose cost is no larger than the optimal solution cost by $w$. 
The low level ECBS also returns a minimum $f_{\mathrm{min}}^i$ value in OPEN for agent $a_i$, which is a lower bound on the cost of the shortest path for $a_i.$ 
The $f_{\mathrm{min}}^i$ value is used to obtain a lower bound at the high level ECBS search. 
At the high level, ECBS prioritizes the CT nodes in the OPEN list 
according to their $lb(N) = \sum_{i=1}^m f_\mathrm{min}^i(N)$ values. 
The high level FOCAL list contains a subset of nodes in OPEN, whose $cost(N) \leq w \, lb(N),$ sorted according to the number of conflicts in $N,$ or $h_c(N).$ 
Since the minimum $lb$ value in OPEN is a lower bound on the optimal solution cost, FOCAL results in a solution whose cost is at most $w\, cost^*.$ 

\subsubsection{Explicit Estimation CBS (EECBS)}

EECBS~\cite{Li2021} improves upon ECBS by addressing the drawbacks of the focal search in the high-level search of ECBS~\cite{Li2021}. 
The number of conflicts in a CT node $N,$ $h_c(N)$ is negatively correlated to the cost of the node $cost(N)$, because when a conflict is resolved in the child node, the cost of the child CT node will likely increase.
Hence, when the focal search chooses a CT node to expand, then the children nodes will not likely be put in the FOCAL list because of their high costs, despite the smaller number of conflicts.
Instead, a neighboring CT node from the FOCAL list with approximately the same number of conflicts will be chosen.
Consequently, the lower bound rarely increases before all nodes in FOCAL are exhausted.
Hence, if the solution is not within the initial suboptimal bound, ECBS may experience difficulty finding a bounded solution~\cite{Li2021}. 

Explicit Estimation Search (EES)~\cite{Thayer2011} addresses this negative correlation issue of the focal search by introducing a third priority queue. The additional priority queue takes the potential cost increment into consideration to remove the negative correlation between OPEN and FOCAL. 
Formally, EES maintains three lists: CLEANUP, OPEN, and FOCAL. 
CLEANUP is a regular A* open list which sorts its elements according to an admissible heuristic $f$. 
OPEN is another regular A* open list that sorts its elements according to a more informed, and possibly inadmissible, heuristic $\hat{f}$. 
FOCAL contains a subset of nodes in OPEN which does not overestimate the minimum $\hat{f}$ value in OPEN by more than a suboptimality factor $w$ and sorts them according to another potentially inadmissible heuristic $d$. 
EES expands a node $best_d$ from FOCAL only if $f(best_d) \leq w f(best_f).$ If not, then EES selects $best_{\hat{f}}$ from OPEN and expands it only if $f(best_{\hat{f}}) \leq w f(best_f)$; otherwise, EES expands $best_f$ from CLEANUP.

EECBS replaces the high-level focal search with EES.
EECBS uses an admissible heuristic $lb=\sum_{i=1}^m f_\mathrm{min}^i$ to sort the CT nodes in CLEANUP.
EECBS also uses an inadmissible heuristic $\hat{f}(N) = cost(N)+\hat{h}(N)$ to sort a CT node $N$ in OPEN, where $\hat{h}(N)$ is some approximated cost increment to resolve the conflicts in $N.$  
EECBS uses an additional inadmissible heuristic $h_c$, the number of conflicts, to sort CT nodes in FOCAL, which contains a subset of OPEN whose $\hat{f}(N)\leq w\, \hat{f}(best_{\hat{f}}).$
EECBS first tries to expand a node $best_{h_c}$ from FOCAL 
with minimum $h_c$ value 
if $cost(best_{h_c}) \leq w\, lb(best_{lb}).$
If not, then EECBS tries to expand a node $best_{\hat{f}}$ from OPEN
with minimum $\hat{f}$ value
if $cost(best_{\hat{f}}) \leq w\, lb(best_{lb}).$
Otherwise, EECBS expands $best_{lb}$ from CLEANUP, increasesing the lower bound. 

The bounded-suboptimality of ECBS/EECBS comes from the fact that a chosen CT node is always bounded above by $w\, lb(best_{lb}),$
where $lb(best_{lb})$ is the lower bound of the optimal solution cost. 
However, the $lb(best_{lb})$ value is not always a tight lower bound of the optimal solution cost, since the value is computed based on the sum $\sum_{i=1}^m f_\mathrm{min}^i,$ where $f_\mathrm{min}^i$ is the lowest $f$-value in OPEN of the low-level search of agent $a_i.$ 
When the low-level focal search expands a goal node from FOCAL and returns a path, $f^i_\mathrm{min}$ may not necessarily be the shortest path length of $a_i.$
Such a gap between $lb(best_{lb})$ and the true optimal solution cost prohibits the high-level search of ECBS/EECBS from exploring the heuristic plans having a lower number of conflicts, i.e., $h_c$, even though those paths may, in fact, be near optimal.



\section{CBS-Budget (CBSB)}


In this work we leverage the colored planning framework~\cite{Wooden2006, Lim2021a, Lim2021b} to tighten 
the gap between the estimated lower bound and the optimal solution cost. 
To this end, 
we propose a new method called budgeted COA* (bCOA*) that finds the shortest path with minimal number of conflicts whose cost is still bounded from above either by a given budget or else it 
finds the shortest path if no such path exists.
bCOA* builds upon our previous Class-Ordered A* (COA*) algorithm~\cite{Lim2021a} which is summarized next.

\subsection{Class-Ordered A* (COA*)}

COA* operates on a weighted colored graph, where the edges are partitioned into a finite set of colors with a total order
and finds the shortest path with minimal inclusion of low-ranked colored edges.
COA* finds such a path by expanding nodes in a similar manner as regular A* but using a different ordering rule to find the newly defined optimal path, that is, the shortest path with minimal inclusion of low-rank colored edges. 

Given a weighted colored graph and a start vertex and a goal vertex, COA* begins the search by initiating a search tree with the start vertex as the root node. 
COA* maintains a single priority queue OPEN, where OPEN prioritizes paths according to their constitute colors and then by their length.
COA* incrementally grows the optimal search tree toward the goal vertex, where a branch of the search tree approximates the shortest path with minimal inclusion of low-rank colored edges from the root to the leaf node. 
In essence, COA* builds a search tree similar to A*, but the search tree of COA* approximates the shortest path with minimal inclusion of low-rank colored edges to the expanded node,
rather than just the shortest path to the expanded nodes.

\begin{figure*}[thpb]
    \def\coaExampleSize{0.18}
	\centering
	\begin{subfigure}{\coaExampleSize\textwidth}
		\includegraphics[width=\myLineScale\linewidth]{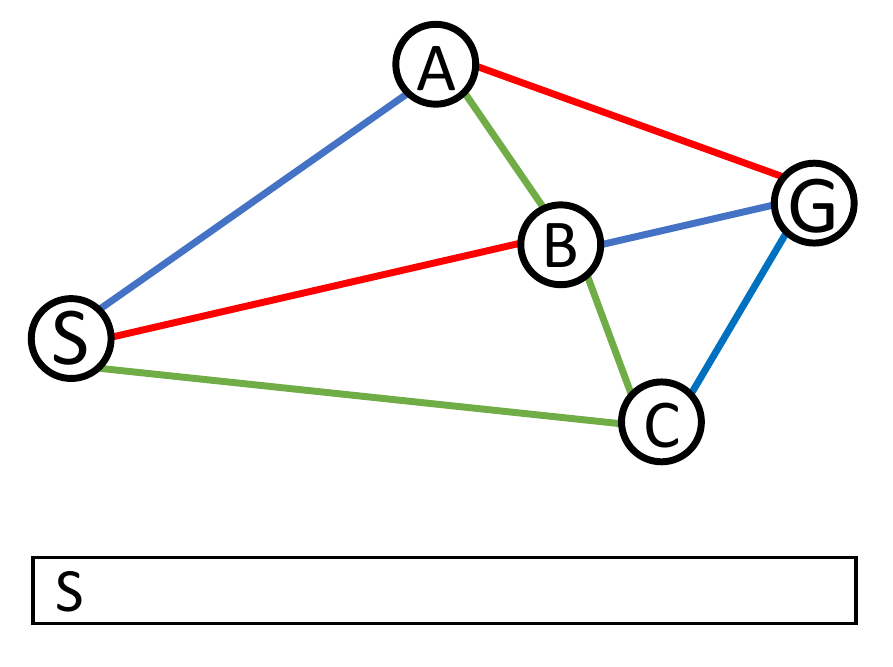}
	    \caption{}
	\end{subfigure}
	\begin{subfigure}{\coaExampleSize\textwidth}
		\includegraphics[width=\myLineScale\linewidth]{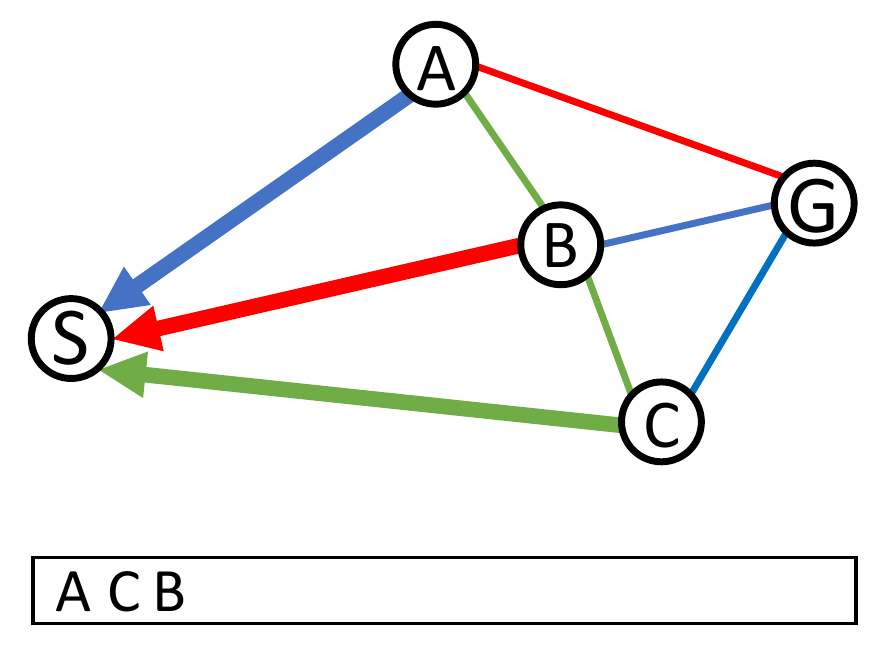}
	    \caption{}
	\end{subfigure}
	\begin{subfigure}{\coaExampleSize\textwidth}
		\includegraphics[width=\myLineScale\linewidth]{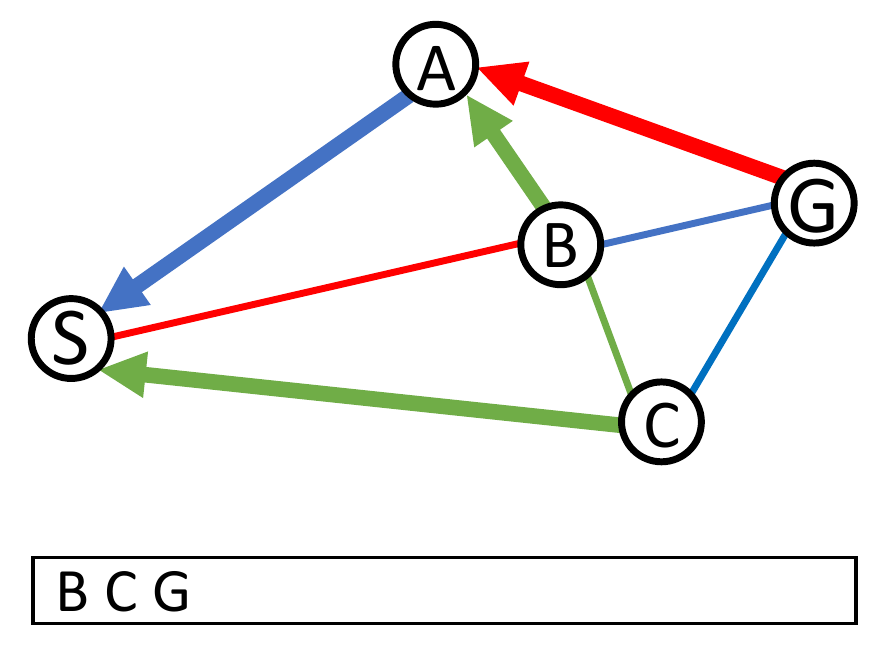}
		\caption{}	
	\end{subfigure}
	\begin{subfigure}{\coaExampleSize\textwidth}
		\includegraphics[width=\myLineScale\linewidth]{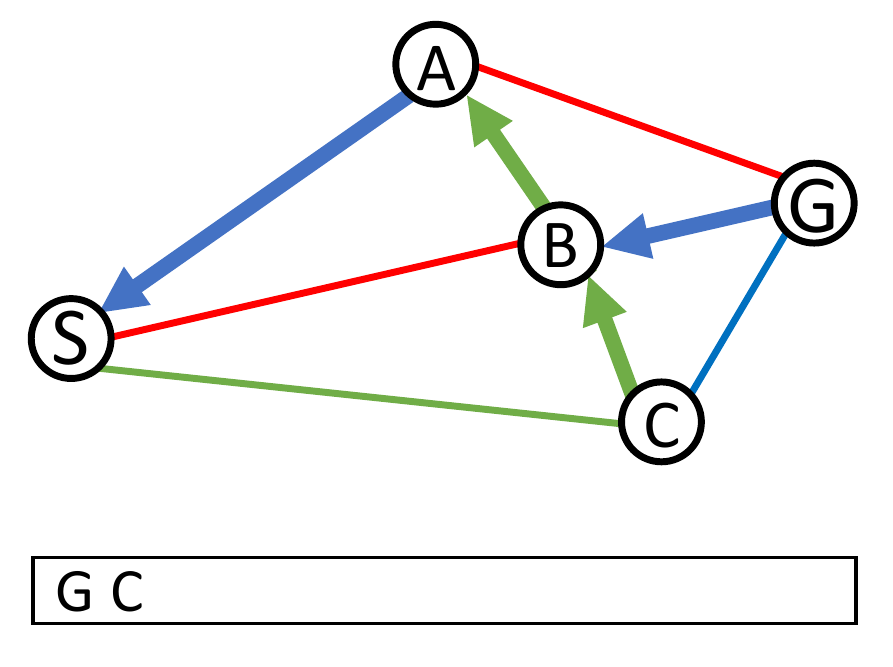}
		\caption{}
	\end{subfigure}
	\begin{subfigure}{\coaExampleSize\textwidth}
		\includegraphics[width=\myLineScale\linewidth]{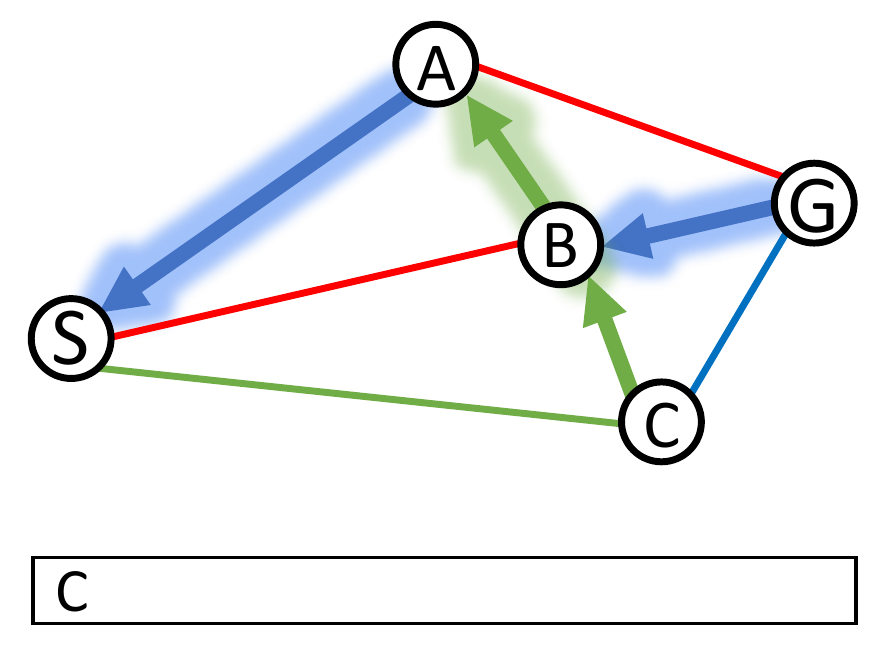}
		\caption{}
	\end{subfigure}
	\caption{Search propagation of COA* on the weighted colored graph from left to right, where COA* incrementally builds an optimal search tree to find the shortest path with minimal inclusion of red, green, and blue edges from S to G (best viewed in color). The arrowed edges are the current search tree in which the arrows are pointing to parent nodes. The box represents the OPEN list of COA* where each letter represents the frontier node of each path explored. COA* propagates the search by expanding the optimal path candidates from the OPEN list.}
	\label{cbsb:f:coastar}
\end{figure*}

Figure~\ref{cbsb:f:coastar} visualizes an example of the COA* search on a weighted colored graph, where the edge set is partitioned into three classes, namely, blue, green, and red. The length of edge is the edge weight. COA* finds the shortest path with minimal inclusion of red, green, and then blue edges from the start vertex S to the goal vertex G.
The search begins by putting the start vertex S to the OPEN list which is shown in the box below the graph (Figure~\ref{cbsb:f:coastar}a).
When S is expanded from the OPEN list, three paths are discovered and they are sorted according to their constitute colors in the OPEN list (Figure~\ref{cbsb:f:coastar}b).
The path to A is expanded next, discovering a better path to B, as the new path contains shorter red edges than the previous path. The search tree is updated and the two newly discovered paths are inserted into the OPEN list (Figure~\ref{cbsb:f:coastar}c).
When the path to B is expanded next, the paths to the two neighboring nodes are updated: the path to C is updated as the new path contains shorter green edges, and the path to G is updated as the new path contains shorter red edges (Figure~\ref{cbsb:f:coastar}d).
When the path to G is expanded, there does not exist any better path to G, halting the search. The optimal path is retrieved by tracing back to the root node (Figure~\ref{cbsb:f:coastar}e).

\subsection{Budgeted COA* (bCOA*)}

bCOA* is a special case of COA*, where expanded paths belong to either one of the two colors: \textit{class}-1 or \textit{class}-2 as follows:
A path belongs to \textit{class}-1 if the path is conflict-free or does not exceed the given budget in terms of its length; 
otherwise the path belongs to \textit{class}-2.  
bCOA* finds the shortest path with minimal inclusion of paths in \textit{class}-2 and then in \textit{class}-1. 

To illustrate the basic idea behind bCOA*, 
consider, for example, the 2D grid world in Figure~\ref{cbsb:f:bcoastar}a where an agent tries to find a path from S to T. 
The black cells are obstacles where no passing paths exists, 
and the blue cell indicates another agent sitting at that location, so a conflict will occur when crossing that cell.
Any path crossing the blue cell will be considered as of being \textit{class}-2.
For simplicity of description, assume $h=0$ for all cells.
bCOA* with an infinite budget (Figure~\ref{cbsb:f:bcoastar}b) builds a search tree from S to T and finds the shortest path that has no conflict.
Now suppose that bCOA* is given a budget $B$=2 (Figure~\ref{cbsb:f:bcoastar}c).
In this case,
bCOA* considers the paths that are more than 2 steps away from S as \textit{class}-2, that is, any path passing through the green cells will be considered as \textit{class}-2. 
As a result, bCOA* with $B$=2 builds a different search tree, and the resulting path goes through the blue cell, because that path is the shortest path with minimal inclusion of \textit{class}-2 paths.
Likewise, when bCOA* is given zero budget, then bCOA* finds the shortest path from S to T
as no \textit{class}-1 paths exist.

\begin{figure}[thpb]
    \def\bcoaExampleSize{0.22}
	\centering
	\begin{subfigure}{\bcoaExampleSize\textwidth}
		\includegraphics[width=\myLineScale\linewidth]{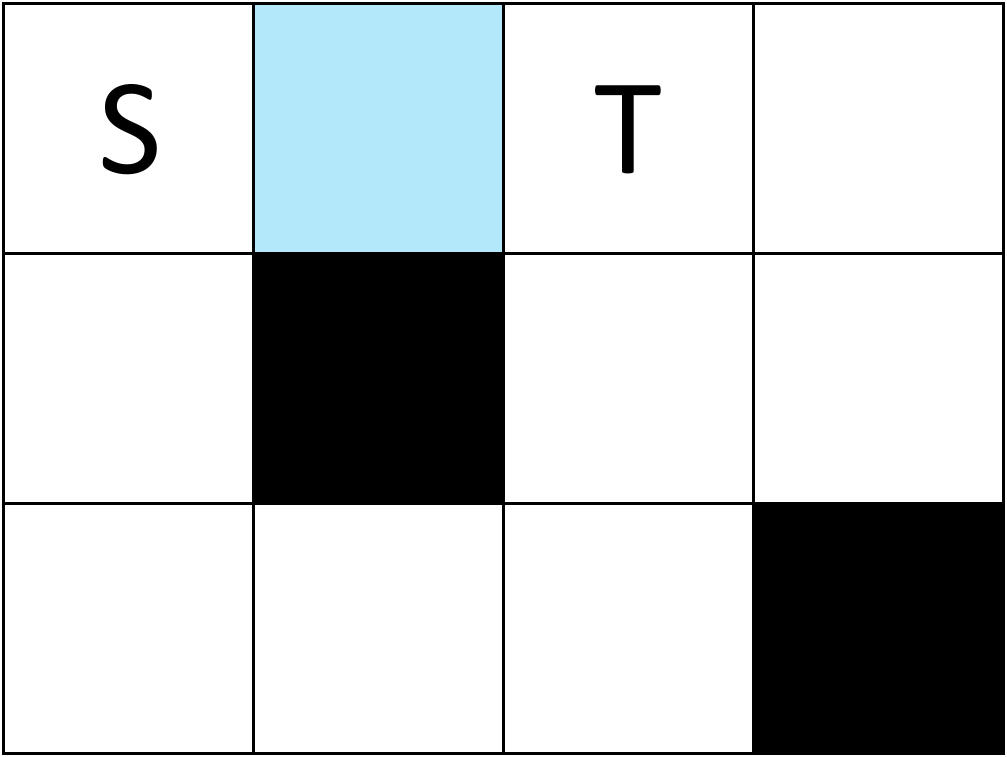}
	    \caption{A grid world}
	\end{subfigure}
	\begin{subfigure}{\bcoaExampleSize\textwidth}
		\includegraphics[width=\myLineScale\linewidth]{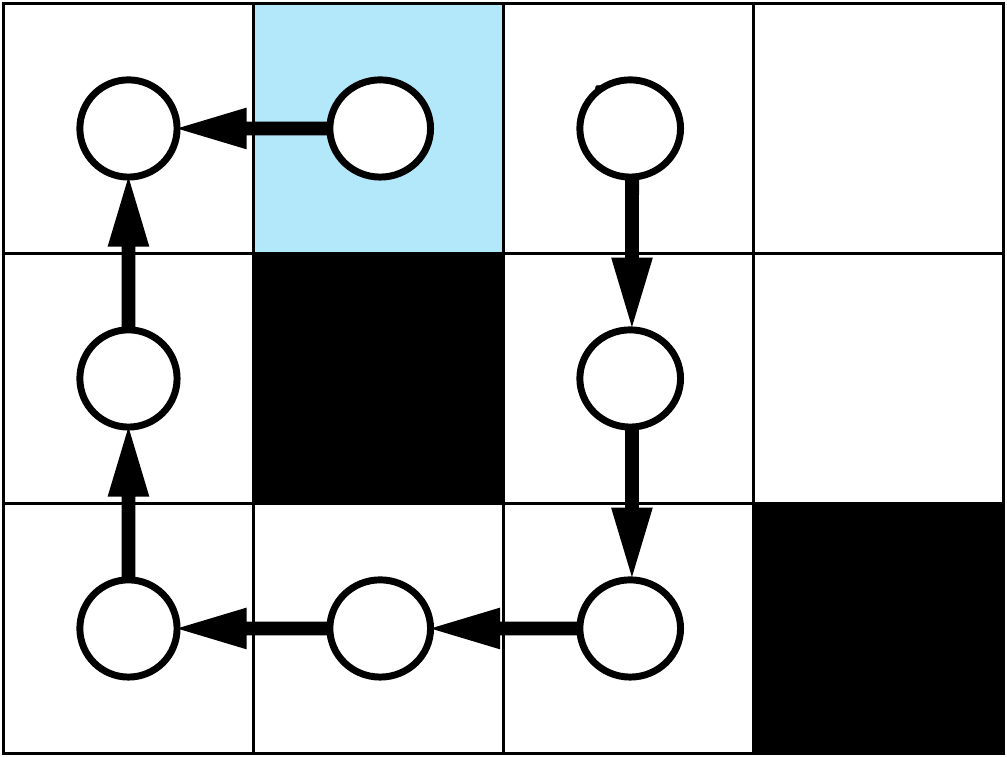}
	    \caption{$B=\infty$}
	\end{subfigure}
	\begin{subfigure}{\bcoaExampleSize\textwidth}
		\includegraphics[width=\myLineScale\linewidth]{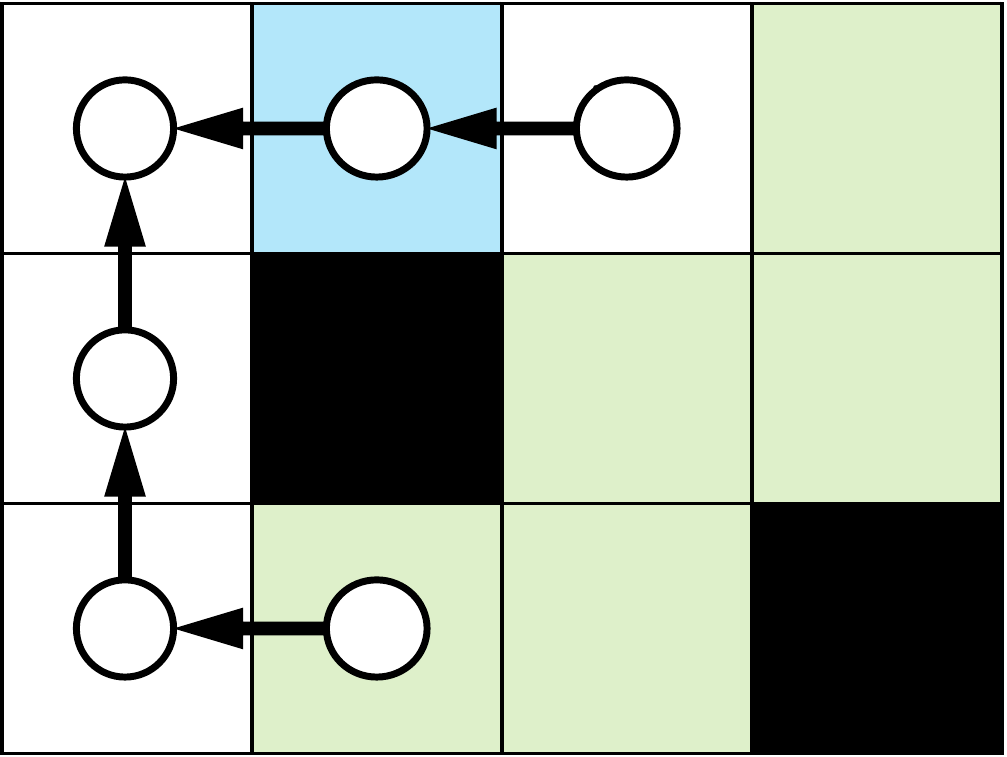}
		\caption{$B$=2}	
	\end{subfigure}
	\begin{subfigure}{\bcoaExampleSize\textwidth}
		\includegraphics[width=\myLineScale\linewidth]{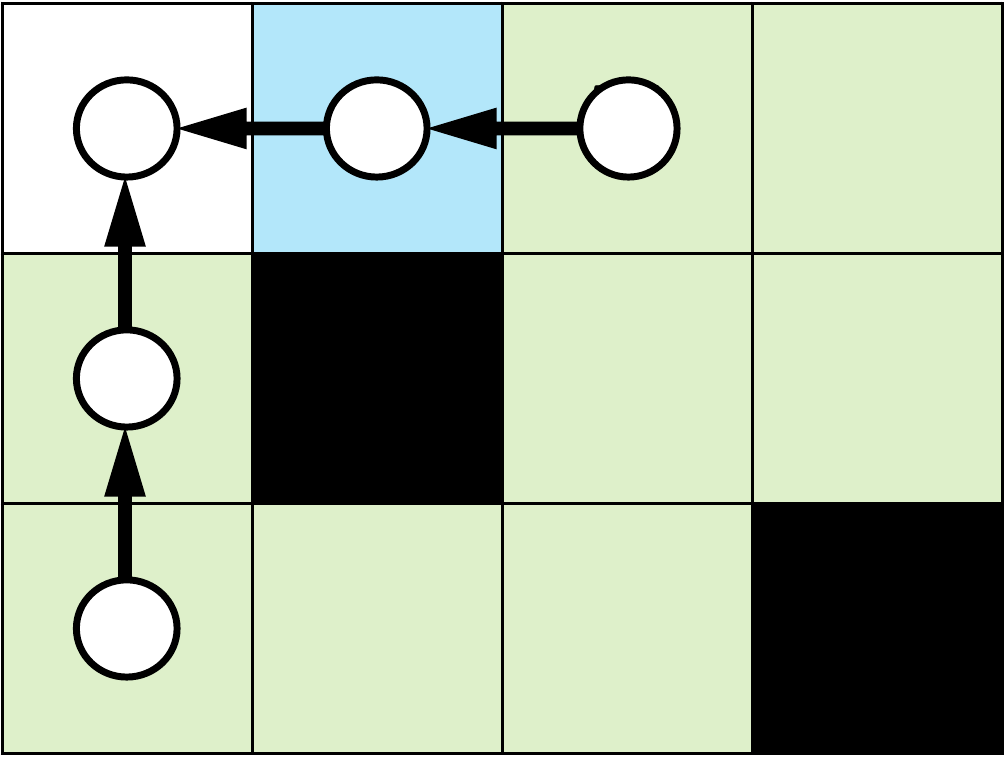}
		\caption{$B$=0}
	\end{subfigure}
	\caption{A 2D grid environment with three bCOA* search trees with different budget values. A path can be retrieved by following the arrows (parent nodes).}
	\label{cbsb:f:bcoastar}
\end{figure}

Given a budget $B$, bCOA* returns a solution having the minimal number of conflicts that is at most $B$-long, if such a path exists. 
If no solution shorter than $B$ exists, then bCOA* returns the shortest path.
Hence the solution returned by bCOA* with a budget $B$ is bounded above by $\max(B,f^*),$ where $f^*$ is the shortest path length in the graph.
Before we prove the properties of bCOA*,
let us first examine how the properties of bCOA* can benefit CBS.  

\subsection{Details of the CBSB algorithm}

The proposed algorithm, CBS-Budget (CBSB), uses the budgeted COA* (bCOA*) on a time-expanded graph (TEG) at the low-level search of CBS and 
uses a modified focal search at the high-level search of CBS,
in order to expand a bounded-suboptimal heuristic plan with the minimum number of conflicts.
The modified high-level focal search of CBSB is different than the high-level focal search of ECBS.
ECBS uses a lower bound, namely, $lb$-value to approximate the suboptimal upper bound, i.e., $w\, lb(best_{lb}) \leq w\, cost^*,$
where $best_{lb}$ is the top CT node of OPEN of the ECBS high-level focal search.
On the other hand, CBSB uses a more informed suboptimal upper bound, called $b$-value which never over-approximates the true optimal solution cost by more than $w$, i.e., $b(best_b)\leq w\,cost^*,$
where $best_b$ is the top node CT of OPEN of the CBSB high-level focal search.

To achieve this more informative suboptimal upper bound, $b(best_b),$  a CT node $N$ of CBSB maintains a list $\set{b_i(N)}_{i=1}^m,$ where $b_i(N)$ is a budget given for agent $a_i$ to find a path using bCOA*.
The $b$-value of $N$ is simply the sum of these individual budgets, i.e., $b(N)=\sum_{i=1}^m b_i(N).$
The root CT node of CBSB is initially assigned $b_i = w\, \hat{f_i},$ where $\hat{f_i}$ is an admissible heuristic cost estimate and $w$ is a user-specified suboptimality factor.
Therefore, $b_i$ does not over-approximate the shortest path of agent $a_i$ with length $f_i^*,$ i.e., $b_i = w\,\hat{f_i} \leq w\, f_i^*.$
These $b_i$ values are incrementally updated based on the result of the low-level search of CBSB, which then are used for the low-level search at the children CT nodes.

At the low level,
bCOA* using the values $b_i$ computes the shortest path with minimal inclusion of conflicts, that is 
at most $b_i$-long for agent $a_i$,
while satisfying the constraints, if one exists;
otherwise, bCOA* finds the shortest path that satisfies the constraints.
Suppose bCOA* finds a path for agent $a_i$ with length $f_i.$
If $b_i < f_i$, then $f_i=f_i^*,$ which we will prove in~Section~\ref{cbsb:subsection:analysis}.
Hence, if $b_i < f_i$ after the low-level search of agent $a_i$, 
CBSB updates $b_i$ for the child CT node with $w\,f_i,$
so that $b_i$ becomes more informed and yet still never over-approximates the shortest path satisfying the current constraints by more than $w$, i.e., $b_i \leq w\,f_i^*.$
Then $b=\sum_i^m{b_i}$ for the current CT node with updated $b_i$'s is computed, which never overestimates the optimal cost subject to the constraints in that CT node.

At the high level, CBSB maintains two priority queues: OPEN and FOCAL. OPEN is a regular A* open list that sorts CT nodes according to their $b$ values. The minimum $b$ value of OPEN, or $b(best_b)$ never overestimates the true optimal solution cost by more than $w$ factor. 
FOCAL contains a subset of OPEN whose cost estimate does not exceed the minimum $b$ of OPEN and sorts according to the number of conflicts, i.e.,
FOCAL$=\set{n\in$~OPEN~$: cost(n)+h(n) \leq b(best_b)}.$
CBSB always chooses a node from FOCAL with the minimum number of conflicts.
Hence, when CBSB expands a goal node $N$, then the cost of the node does not exceed the optimal solution cost by more than a $w$ factor, i.e., $cost(N) \leq b(best_b)\leq w\,cost^*.$
Note that in case $w=1,$ the $b$-value of each CT node is equal to the $cost$-value since $b_i=f_i^*$ for all agents, so FOCAL of CBSB always expands the same node as head(OPEN), and CBSB reduces down to CBS.

\begin{algorithm}
	\caption{High-level search of CBSB\textcolor{blue}{-BP}}\label{cbsb:a:high-level}
	\begin{small}
	\begin{algorithmic}[1]
	    \Procedure{PushNode}{$N$}
	        \State OPEN $\gets$ OPEN $\cup$ $\set{N}$;
	        \If{$cost(N) + h(N) \leq b$(head(OPEN))}
	            \State FOCAL $\gets$ FOCAL $\cup$ $\set{N}$;
	        \EndIf
	    \EndProcedure	    
	    \Procedure{SelectNode}{$ $}
	        \If{$b_\mathrm{min} <$ $b$(head(OPEN))}
	            \State $b_\mathrm{min} \gets$ $b$(head(OPEN));
	            \For{\textbf{all } $N\in$ OPEN}
	                \If{$cost(N)+h(N) \leq b_\mathrm{min}$}
	                    \State FOCAL $\gets$ FOCAL $\cup$ $\set N$;
	                \EndIf 
	            \EndFor 
	        \EndIf
	        \State $N\gets$ head(FOCAL);
	        \State FOCAL $\gets $ FOCAL $\backslash \set{N}$;
	        \State OPEN $\gets $ OPEN $\backslash \set{N}$;
	        \State \Return $N;$
	    \EndProcedure	    
		\Procedure{Main}{$ $}
            \State $R\gets$\textsc{GenerateRoot}();
            \State \textsc{PushNode}($R$);
            \While{OPEN is not empty}
                \State $N\gets$\textsc{SelectNode}();
                \If{$N.conflicts$ is empty} \label{cbsb:al:terminate}
                    \State \Return $N.paths$;
                \EndIf
                \State $conflict\gets$ \textsc{ChooseConflict}($N$); \label{cbsb:a:chooseconflict}
                \State $constraints\gets$ \textsc{ResolveConflict}($conflict$);
                \label{cbsb:a:resolveconflict}
                \For{\textbf{each } $constraint$ in $constraints$}
                    \State $N'\gets$ \textsc{GenerateChild}($N, constraint$);
                    \If{\textcolor{blue}{$cost(N')\leq b_\mathrm{min}$ and $h_c(N')<h_c(N)$}} \label{cbsb:al:bypass-begin}
                        \State \textcolor{blue}{$N.paths \gets N'.paths$;}
                        \State \textcolor{blue}{$N.conflicts \gets N'.conflicts$;}
                        \State \textcolor{blue}{Go to Line~\ref{cbsb:al:terminate};} \label{cbsb:al:bypass-end}
                    \EndIf
                    \State \textsc{PushNode}($N'$); 
                \EndFor 
            \EndWhile
		\EndProcedure
	\end{algorithmic}
	\end{small}
\end{algorithm}

The pseudo code of the high-level CBSB is provided in Algorithm~\ref{cbsb:a:high-level}.
CBSB is similar to CBS in the way it chooses a conflict and imposes constraints on the two children CT nodes (Algorithm~\ref{cbsb:a:high-level} Line~\ref{cbsb:a:chooseconflict}-\ref{cbsb:a:resolveconflict}). 
CBSB is different from CBS in the way a CT node is prioritized for expansion.
When a new CT node is generated, CBSB pushes the node to OPEN and also to FOCAL if the cost estimate is lower than the minimum $b$-value in OPEN.
When choosing a CT node for expansion, CBSB first updates FOCAL according to the current minimum $b$-value in OPEN, and then selects a node from FOCAL with the minimum number of conflicts.
CBSB repeats until the chosen node contains no conflict.

\begin{algorithm}
	\caption{Low-level search of CBSB}\label{cbsb:a:low-level}
	\begin{small}
	\begin{algorithmic}[1]
	    \Procedure{GenerateRoot}{$ $}
	        \State $R\gets$ create a new CT node
			\For {\textbf{all } agent $a_i$}
				\State $R.b_i\gets w\,\hat{f_i}$;
				\State $R.p_i\gets$ \textsc{FindPath($a_i, R.b_i$)};
				\If{$R.b_i < |R.p_i|$}
				    $R.b_i \gets w\, |R.p_i|$;
				\EndIf
			\EndFor
			\State $R.b\gets \sum_i^m{R.b_i}$
			\State \Return $R$;
	    \EndProcedure
	    \Procedure{GenerateChild}{$N, constraint$}
	        \State $N' \gets $ copy $N$;
            \State $N'.constraints$ $\gets$ $N'.constraints \cup constraint$;
            \State $N'.p_k\gets$ \textsc{FindPath}($a_k, N'.b_k, N'.constraints$);
            \If{$N'.b_k < |N'.p_k|$}
                $N'.b_k \gets w\, |N'.p_k|$;
            \EndIf
            \State $N'.b\gets \sum_i^m{N'.b_i}$;
	        \State \Return $N'$;
	    \EndProcedure	    
	\end{algorithmic}
	\end{small}
\end{algorithm}

\subsection{CBSB with Bypassing Conflicts}
Bypassing Conflict (BP)~\cite{Boyarski2015b} is a technique that can be used on top of CBSB (Algorithm~\ref{cbsb:a:high-level} Line \ref{cbsb:al:bypass-begin}-\ref{cbsb:al:bypass-end}) to reduce the size of the CT by avoiding 
the splitting of CT nodes. 
This improvement is denoted as CBSB-BP and is shown in blue in Algorithm~\ref{cbsb:a:high-level}.
Since CBSB finds a bounded suboptimal path, we can relax the conditions of accepting bypasses. 
Specifically, when expanding a CT node $N$ and generating a child CT node $N',$ CBSB-BP checks if $cost(N')<b_\mathrm{min}$ and $h_c(N')<h_c(N)$, where $b_\mathrm{min}$ is the $b$-value of the top node of OPEN and $h_c$ is the number of conflicts. 
Note that $cost(N')<b_\mathrm{min}$ ensures that the paths in $N'$ is still bounded suboptimal.
If so, then CBSB-BP adopts the paths of $N'$ and the conflicts of $N'$ for $N$.
Otherwise, CBSB-BP splits $N$ as before. 

\subsection{Example}
Let us momentarily, skip the lines in blue (Algorithm~\ref{cbsb:a:high-level} ~Line~\ref{cbsb:al:bypass-begin}-\ref{cbsb:al:bypass-end}) and let us describe the basics of CBSB using the example shown in Figure~\ref{cbsb:f:cbsb-example-env}.
Figure~\ref{cbsb:f:cbsb-example-env} shows an MAPF instance where the bears need to swap their positions, that is, agent 1 at A needs to reach C, and agent 2 at C needs to reach A. 
\begin{figure}[thpb]
	\centering
		\includegraphics[width=0.7\linewidth]{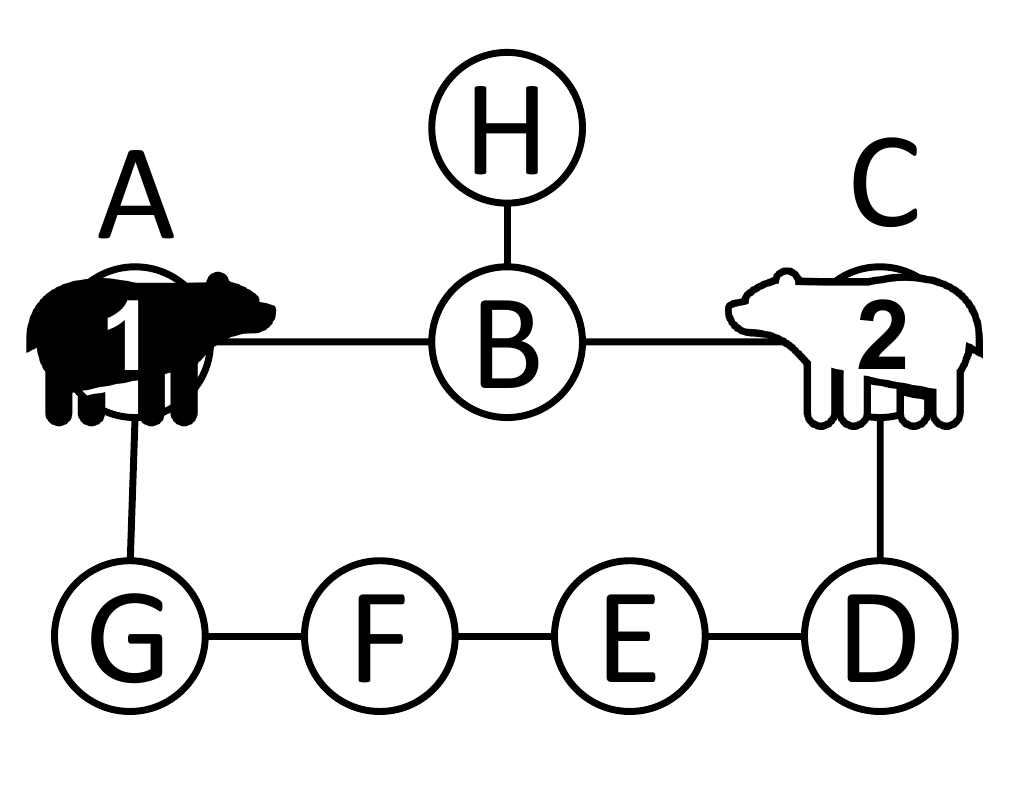} 
	\caption{An example of MAPF instance with 2 agents. The two bears must swap their positions.}
    \label{cbsb:f:cbsb-example-env}
\end{figure}

\begin{figure*}[t]
	\centering
	\begin{subfigure}{0.6\textwidth}
	    \includegraphics[width=1\linewidth]{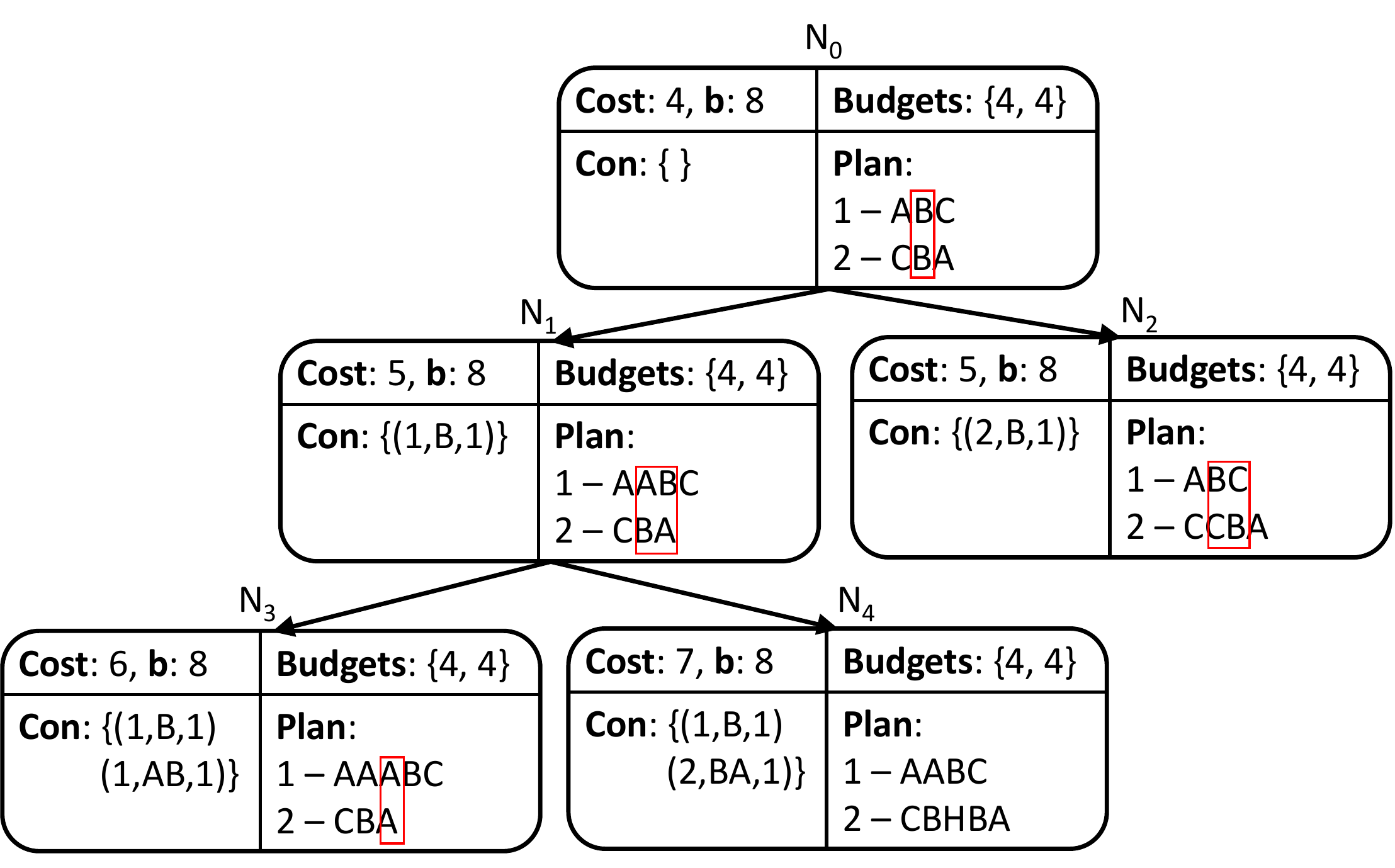} 
	    \caption{The CT of CBSB with $w$=2.}
	    \label{cbsb:f:cbsb-example-w2-ct}
	\end{subfigure}\\[1ex]	
    \begin{subfigure}{0.6\textwidth}
        \includegraphics[width=1\linewidth]{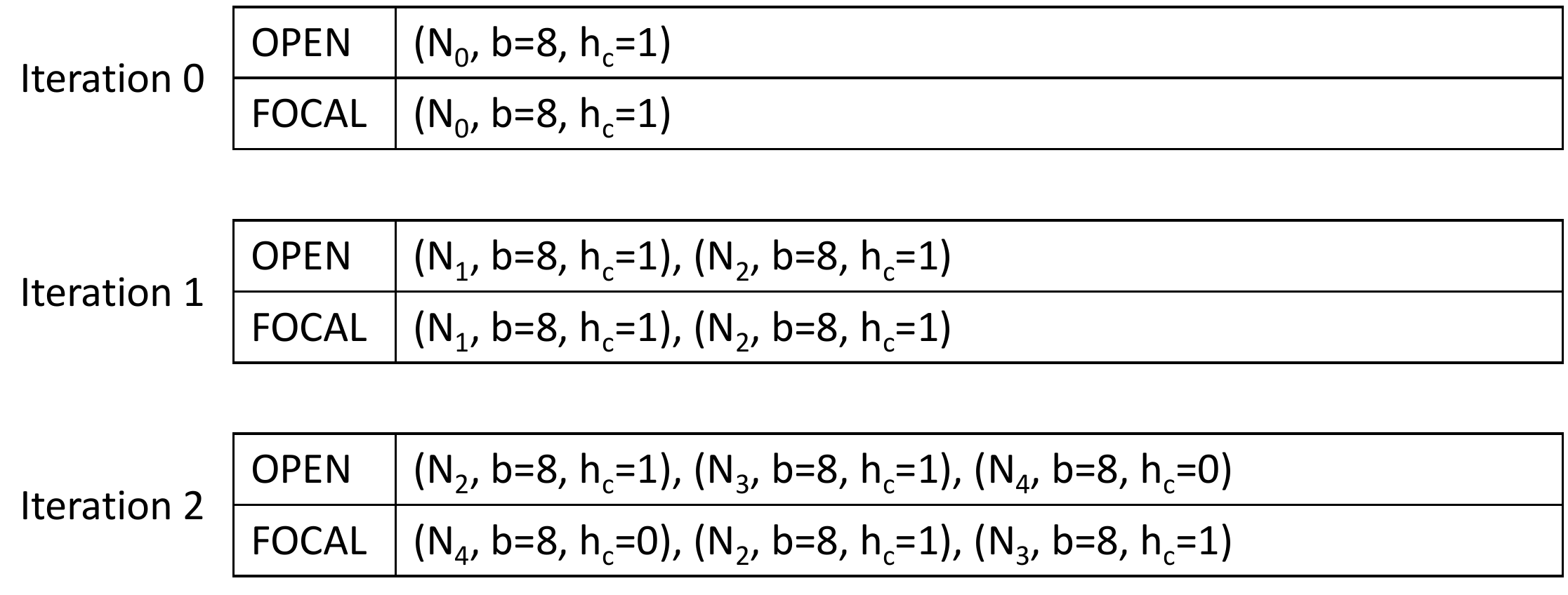} 
	    \caption{The high level CBSB search iterations with $w=2$.}
	    \label{cbsb:f:cbsb-example-w2-queue}
    \end{subfigure}
    \caption{The CBSB search with $w=2.$}
    \label{cbsb:f:cbsb-example-w2}
\end{figure*}

We will first consider the CBSB search with suboptimality factor of 2 (i.e., $w$=2). 
The corresponding CT is shown in Figure~\ref{cbsb:f:cbsb-example-w2-ct}. 
The OPEN and FOCAL lists of CBSB of each iteration are shown in Figure~\ref{cbsb:f:cbsb-example-w2-queue}.
The search begins by generating the root node N$_0$, where each agent is given an initial budget $b_i$ of 4 according to its admissible heuristic path length and the given suboptimality factor, i.e., $b_i=w\,\hat{f_i}.$ bCOA* then finds the path with the given budget. Since the path found for each agent does not exceed the given budget, all the budgets stay the same after the low-level searches. The root node is inserted into OPEN and FOCAL.
OPEN sorts CT nodes according to their $b$-value which is the sum of the individual budgets. FOCAL sorts the subset of the OPEN list whose $cost$ does not exceed the minimum $b$, according to the number of conflicts, $h_c.$

During the first iteration, CBSB chooses the top node from FOCAL, namely N$_0$.
Since there exists a conflict at vertex B at timestep 1 in N$_0$, CBSB generates two children nodes to resolve the conflict.
The children CT nodes N$_1$ and N$_2$ are generated by replanning for the conflicting agents with additional constraints to prevent the conflict. 
When an agent replans, the existing paths of other agents are taken into consideration such that bCOA* finds the shortest path with minimal inclusion of conflicts or paths exceeding the given budget in length. 
At the CT node N$_1,$ the agent 1 finds a new path using the budget of agent 1 at N$_0,$
and at the CT node of N$_2,$ the agent 2 finds a new path using the budget of agent 2 at N$_0.$
Since these new path lengths do not exceed the initial budgets, the budgets stay the same after the searches. 
The generated CT nodes N$_1$ and N$_2$ are put in OPEN and FOCAL accordingly. 

At the next iteration, the CT node N$_1$ is chosen from FOCAL for expansion. The two children nodes N$_3$ and N$_4$ are generated, where each node contains an additional constraint to resolve the chosen conflict in N$_1.$ 
The budgets stay the same after the low-level searches as before. 
The generated CT nodes are put in OPEN and FOCAL accordingly.
Since the cost of N$_4$ is less than the minimum $b$-value in OPEN, N$_4$ is put in FOCAL. Also, since N$_4$ has 0 conflict, the node becomes the head of FOCAL. 
Hence, at the next iteration N$_4$ is chosen and the search halts, returning the paths in N$_4$ as a solution.
The paths in N$_4$ are guaranteed to be bounded suboptimal, that is, the cost does not exceed twice that of the optimal solution length.
In fact, the paths in N$_4$ is the optimal solution to this MAPF instance. 

\begin{figure*}[th]
	\centering
		\includegraphics[width=1\linewidth]{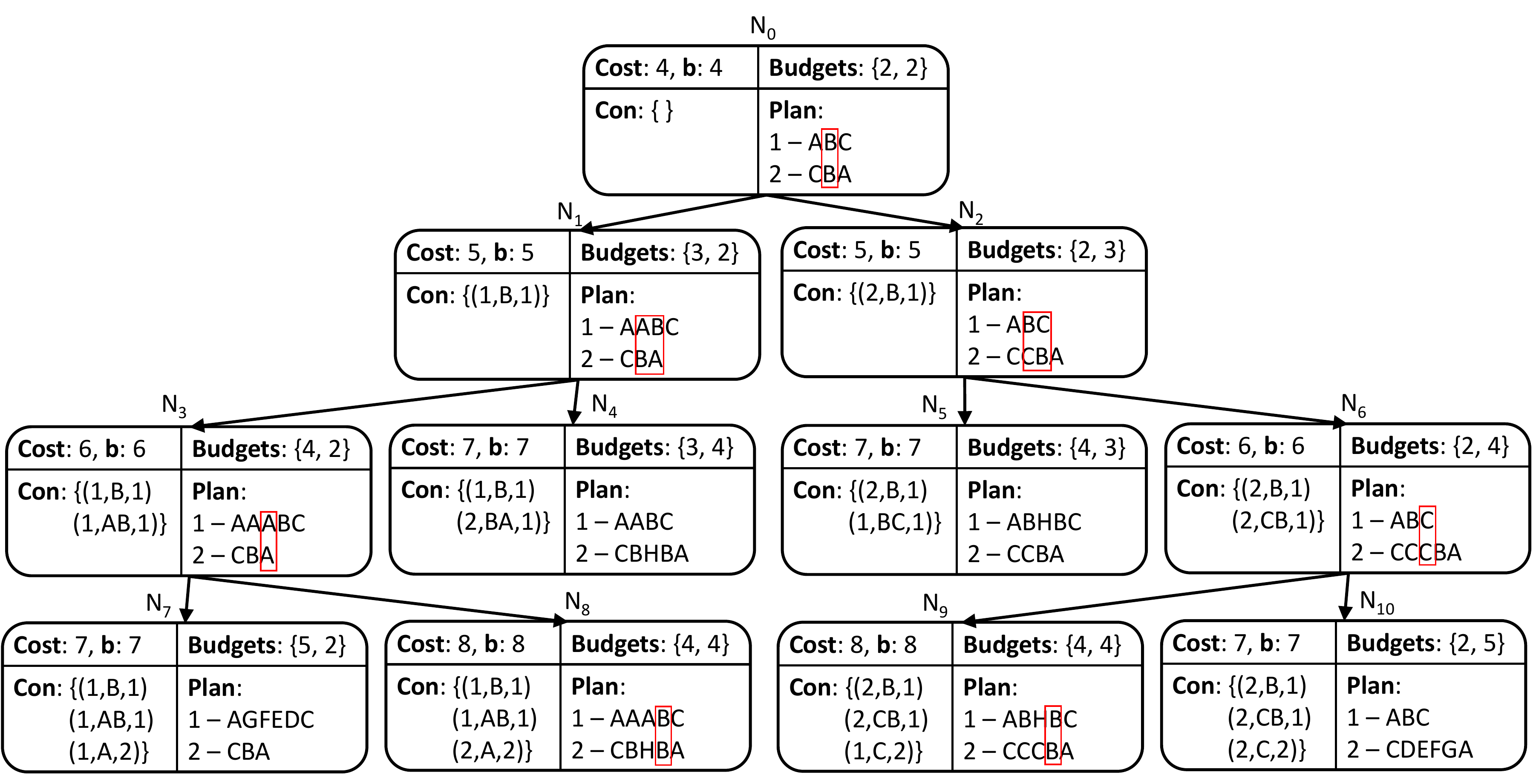} 
	\caption{The CT of CBSB with $w$=1.}
    \label{cbsb:f:cbsb-example-w1-ct}
\end{figure*}

\begin{figure*}[th]
	\centering
		\includegraphics[width=.8\linewidth]{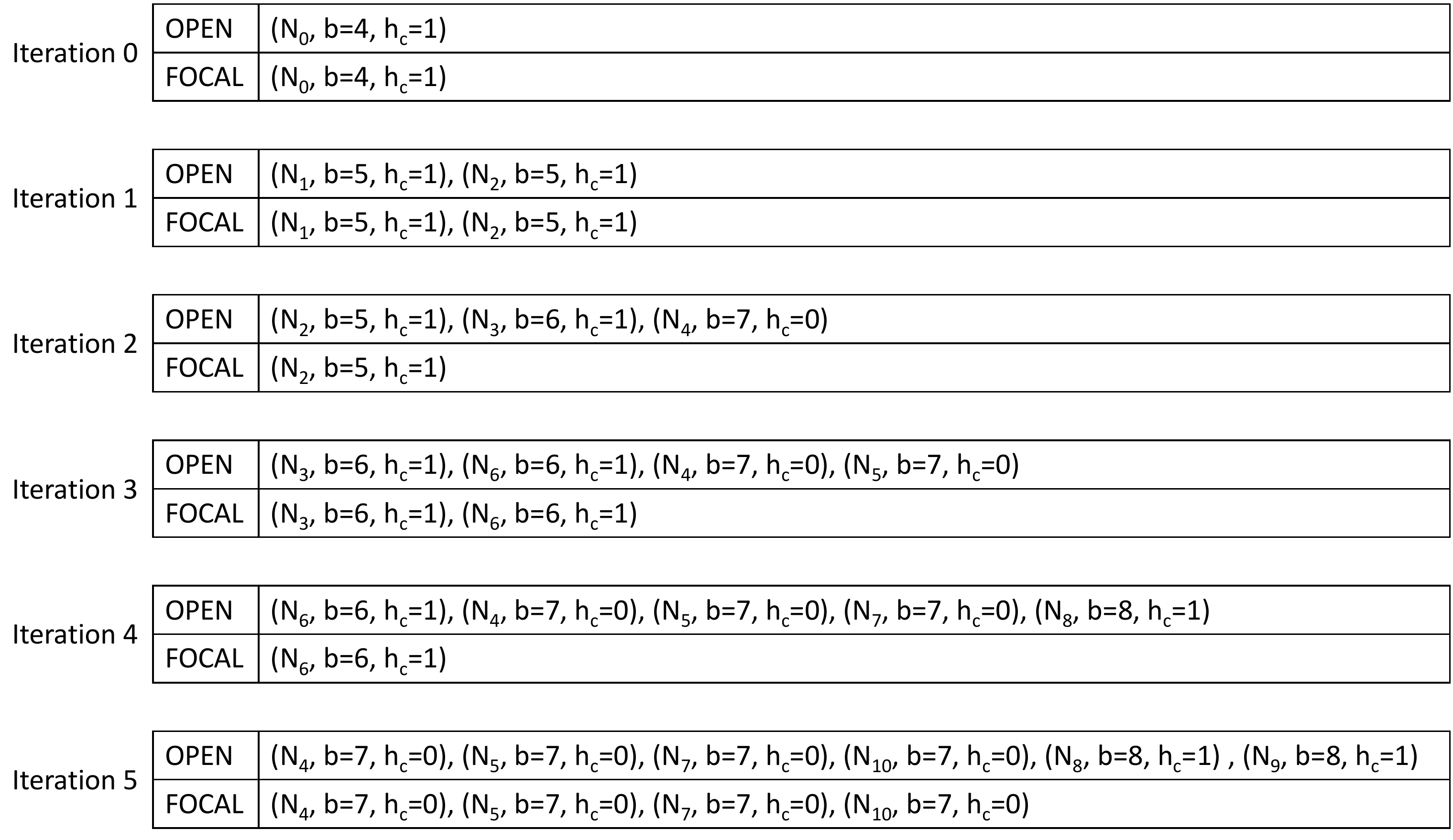} 
	\caption{The high level CBSB search iterations with $w=1$.}
    \label{cbsb:f:cbsb-example-w1-queue}
\end{figure*}

Now let us consider solving the same problem but with $w=1,$ that is, we wish to find the optimal solution. 
The purpose of this example is to show that CBSB with $w=1$ is equivalent to CBS.
The corresponding CT is shown in Figure~\ref{cbsb:f:cbsb-example-w1-ct}. 
The OPEN and FOCAL lists of CBSB of each iteration is shown in Figure~\ref{cbsb:f:cbsb-example-w1-queue}.
The search begins by generating the root node N$_0$, where each agent is given an initial budget $b_i$ 2 according to its admissible heuristic path length and the given suboptimality factor, i.e., $b_i=w\,\hat{f_i}$ and finds the path with bCOA* with the given budget. Since the path found for each agent does not exceed the given budget, all the budgets stay the same after the low-level searches at the root. The root node is inserted into OPEN and FOCAL.


During the first iteration, the CT node N$_0$ is chosen from FOCAL and expanded to resolve the conflict found at vertex B at timestep 1. 
At the CT node N$_1,$ agent~1 finds a new path with the budget of agent~1 at N$_0,$ and since the path returned (AABC) of length 3 is longer than the given budget of 2, the budget of agent~1 is updated to 3. 
Likewise, agent 2 at the CT node N$_2$ returns a path longer than the given budget, increasing the budget accordingly. 
The CT nodes N$_1$ and N$_2$ are put in OPEN and FOCAL.
The CT node N$_1$ is chosen at next iteration (Figure~\ref{cbsb:f:cbsb-example-w1-queue} at Iteration 2). Two children nodes N$_3$ and N$_4$ are generated and put in OPEN with their updated $b$ values. Although N$_4$ is conflict free it is not put in FOCAL, since its $cost$-value exceeds the minimum $b$-value in OPEN. Consequently, the CT node N$_2$ is chosen from FOCAL for expansion at iteration 3.  
Likewise, when N$_2$ is chosen for expansion, N$_5$ and N$_6$ are generated and these children CT nodes are put in OPEN. Similar to N$_4,$ N$_5$ is not put in FOCAL as the cost exceeds the minimum $b$-value in OPEN. 
The goal nodes N$_4$ or N$_5$ are not chosen until the minimum $b$-value increases later in the search. 
This is a well studied pathological case where CBS is inefficient~\cite{Sharon2015, Li2021b} that is, when there exist many infeasible plans due to conflicts whose cost is yet lower than the optimal solution.
Eventually, N$_4$ is chosen and the search halts, confirming that the paths in N$_4$ is the optimal solution.
Note that CBSB with suboptimality factor of 1 ($w=1$) is equivalent to CBS in the way a CT node is chosen for expansion, because the $cost$-value and the $b$-value of a CT node is always the same when $w=1.$

\section{Analysis of CBSB}

In this section, we analyze the properties of CBSB. 
First, we examine the properties of the bCOA* algorithm, and then
we analyze the completeness and bounded suboptimality of CBSB. 

\subsection{Properties of bCOA*} \label{cbsb:subsection:analysis}

Let $\Pi$ be the set of all paths in graph $G$,
let $\Pi_B\subseteq \Pi$ be the subset of paths that are not longer than $B$, and
let $P_B\subseteq \Pi_B$ be the subset of paths that are conflict-free. 
We call an element of $P_B$ a \textit{class}-1 path 
and an element of the complement of $P_B$ a \textit{class}-2 path.
bCOA* finds the shortest path with minimal inclusion of \textit{class}-2 path and then \textit{class}-1 path.
We first examine the properties of a general class of such algorithms, denoted by Y($B$), 
that find the shortest path with minimal inclusion of paths either in conflict 
or longer than $B$-timesteps. 
Note that bCOA* is an instance of Y($B$). 
We prove the properties Y($B$).


\begin{restatable}{lemma}{BCOAUpperBound}\label{lemma:bcoa:upper_bound}
	Y($B$) finds a path having length at most $B$, if such a path exists. 
\end{restatable}

\begin{proof}
	Let $\Pi$ be the set of all paths. 
	Let $\Pi_B\subseteq \Pi$ be the set of all paths whose length does not exceed $B.$ 
	By assertion, $\Pi_B \neq \varnothing.$
	Let $P_B \subseteq \Pi_B$ be the subset of all conflict-free paths in $\Pi_B.$ 
	Denote by $P_C = \Pi_{B} \cap P_B^\mathsf{c},$ the subset of all conflicting paths in $\Pi_B.$
	Then $\Pi = P_B \cup P_C \cup \Pi_B^\mathsf{c},$ where $P_B,$ $P_C,$ and $\Pi_B^\mathsf{c}$ are disjoint. 
	If $P_B\neq \varnothing$, then Y($B$) finds a path in $P_B.$
	Otherwise, Y($B$) finds the shortest path in $P_C \cup \Pi_B^\mathsf{c}.$ But $P_C$ is not empty since $\Pi_B \neq \varnothing$ and $P_B = \varnothing$. 
	Hence, Y($B$) will find a path in $P_C$ since any path in $\Pi_B^\mathsf{c}$ is longer in length than all paths in $P_C.$
\end{proof}

The next lemma states that if the given budget is equal to the shortest path length, then Y($B$) finds the shortest path. 
Also, Y($B$) finds the conflict-free shortest path, if one exists. 

\begin{restatable}{lemma}{BCOABoundary}\label{lemma:bcoa:boundary}
	If $B=f^*,$ the shortest path length in the graph, then Y($B$) finds the shortest path with minimal inclusion of conflicts.
\end{restatable}

\begin{proof}
	Let $\pi^*$ be the shortest path in the graph.
	If $\pi^*$ is conflict-free, then Y($B$) finds $\pi^*$ since it is the shortest path without conflicts. 
	Otherwise, Y($B$) finds $\pi^*$ since it is the shortest path among the set of paths either in conflict or longer than $B.$ 	
\end{proof}

The next lemma shows that if the given budget is less than the shortest path length, then Y($B$) finds the shortest path. 
\begin{restatable}{lemma}{BCOALowerBound}\label{lemma:bcoa:lower_bound}
	If $B<f^*,$ the shortest path length in the graph, then Y($B$) finds the shortest path with length $f^*$.
\end{restatable}

\begin{proof}
	Any path in the graph will be longer than $B$. Hence Y($B$) finds the shortest path among all paths in the graph.
\end{proof}

Given these lemmas, we have the following corollaries.
\begin{restatable}{corollary}{BCOAsufficiency}\label{corollary:bcoa:sufficiency}
	If $B\leq f^*$, the shortest path length in the graph, then Y($B$) finds the shortest path.
\end{restatable}

\begin{proof}
	The result simply follows from Lemma~\ref{lemma:bcoa:boundary} and Lemma~\ref{lemma:bcoa:lower_bound}.
\end{proof}

Corollary~\ref{corollary:bcoa:sufficiency} gives a sufficient condition so that, 
given that a budget is less than or equal to the shortest path cost, 
then the algorithm Y finds the shortest path. 
The next corollary gives a necessary condition so that if the path length found by the algorithm Y is shorter than the given budget, then the path is necessarily 
the shortest path. 
This corollary is useful since the $f^*$-values are not known a priori. 

\begin{restatable}{corollary}{BCOANecessity} \label{corollary:bcoa:necessity}
	If Y($B$) finds a path of length $C$, where $B<C,$ then $C=f^*$ the shortest path length.
\end{restatable}

\begin{proof}
	Suppose, ad absurdum, that $f^* < C.$ If $B\leq f^*$ then Y($B$) finds the path of length $f^*$ by Corollary~\ref{corollary:bcoa:sufficiency}, contradicting that Y($B$) found a path of length $C > f^*.$
	So it must be that $B > f^*.$ Since $B$ is sufficiently large, by Lemma~\ref{lemma:bcoa:upper_bound} Y($B$) finds a path that is at most $B$-long. 
	But Y($B$) found a path of length $C>B,$ leading to a contradiction. 
	Hence $f^* = C.$
\end{proof}

Next, we examine how different budget values affect path quality.

\begin{restatable}{corollary}{BCOABudgetLength}\label{corollary:bcoa:budget_length}
	If $B_1 < B_2$, then Y($B_1$) finds a path that is not longer in length than Y($B_2$).
\end{restatable}

\begin{proof}
	Let $f^*$ be the shortest path length in the graph. 
	We consider three cases:
	\begin{enumerate}[leftmargin=*]
		\item $B_1 < B_2 \leq f^*:$ Both Y($B_1$) and Y($B_2$) find the shortest path with length $f^*$ by Corollary~\ref{corollary:bcoa:sufficiency}.
		\item $B_1 \leq f^* < B_2:$ Y($B_1$) finds the shortest path by Corollary~\ref{corollary:bcoa:sufficiency}.
		\item $f^* < B_1 < B_2:$ Let $p_1$ and $p_2$ be the paths found by Y($B_1$) and Y($B_2$), respectively. 
		Suppose, ad absurdum, that $p_2$ is shorter than $p_1.$
		If $p_2$ is conflict-free, then Y($B_1$) would have found $p_2$ instead of $p_1$, since $p_2$ is shorter. 
		So $p_2$ must be in conflict. 
		Now, if $p_1$ is in conflict, then Y($B_1$) would have found $p_2$ instead of $p_1$, since $p_2$ is shorter. 
		So $p_1$ must be conflict-free.
		But then Y($B_2$) would have found $p_1$ instead of $p_2$, since $p_1$ is conflict-free and shorter than $B_2,$ contradicting  the fact that Y($B_2$) found $p_2$.
		Hence, Y($B_2$) cannot find a shorter path than Y($B_1$).    
	\end{enumerate}
	In all three cases, Y($B_1$) finds a path that is no longer than a path found by Y($B_2$).
\end{proof}

The following corollary states that the algorithm Y with a lower budget finds no fewer conflicts than when using a higher budget. 

\begin{restatable}{corollary}{BCOABudgetConflict}\label{corollary:bcoa:budget_conflict}
	If $B_1 < B_2$, then Y($B_1$) finds a path that has no fewer conflicts than Y($B_2$).
\end{restatable}

\begin{proof}
	Let $p_1$ and $p_2$ be the paths found by Y($B_1$) and Y($B_2$), respectively.
	Suppose $p_1$ has fewer conflicts than $p_2.$ 
	By Corollary~\ref{corollary:bcoa:budget_length}, $p_1$ is not longer than $p_2.$ 
	Hence, Y($B_2$) would have found $p_1$ instead of $p_2$, a contradiction.
\end{proof}

Note that Y(0) is equivalent to A*, as Y(0) always finds the shortest path regardless of the path class, and Y($\infty$) is equivalent to regular COA*~\cite{Lim2021a}, as it always finds the shortest path with minimal conflicts. 
\subsection{Properties of CBSB}

Given the previous properties of the algorithm Y, we examine next the properties of the main algorithm CBSB, which uses Y($b_i$) (equivalently, bCOA*) with a budget $b_i$ at the low level search for agent $a_i$ and a focal search at the high level search of CBS.
We first prove that for a finite $b=\sum_i^m b_i$, there is a finite number of CT nodes. Then, we prove that CBSB must finish before exhausting all the CT nodes. 

\begin{restatable}{theorem}{CBSBFiniteNodes}
	\label{theorem:cbsy:finite_nodes}
	For a finite $b,$ there is a finite number of CT nodes.
\end{restatable}
\begin{proof}
    Since $b=\sum_i^m b_i$ is finite, $b_i$ is finite for all $i=1,\ldots,m$. 
	Let $f_i^*$ be the shortest path of agent $a_i.$
	Then, bCOA* finds a path for agent $a_i$ that has length at most $M_i = \max\set{b_i, f_i^*}.$
	Let $M = \max_i{M_i},$ then
	no CT node contains a path longer than $M$.
	Hence, no conflict can occur after $M$-timestep and no constraint after $M$-timestep can be generated. 
	Since there is a finite number of such constraints 
	(at most $M  |V| k$ for vertices and $M |E| k$ for edges), there is also a finite number of CT nodes that contain such constraints. 
\end{proof}

\begin{restatable}{theorem}{CBSBCompletness}   	\label{theorem:cbsy:completness}
	CBSB returns a solution, if one exists. 
\end{restatable}

\begin{proof}
    CBSB expands all nodes whose cost does not exceed the minimum $b$-value.
    The $b$-value of the CT nodes are monotonically increasing.
    Since for each $b$ there is a finite number of CT nodes, 
    a solution whose cost does not exceed $b$ must be found after expanding a finite number of CT nodes. 
\end{proof}

Finally, we prove that a solution returned by CBSB is bounded suboptimal. 
First, we show that the minimum $b$-value in OPEN does not overestimate the optimal solution cost by more than $w,$ a user-specified suboptimality factor. 
Hence, when a goal CT node is expanded whose cost does not exceed the minimum $b$-value in OPEN, then the goal node will contain a bounded suboptimal solution.

\begin{restatable}{theorem}{CBSBminimumB}   	\label{theorem:cbsy:minimumB}
    A CT node in OPEN with minimum $b$-value does not over-approximate the optimal solution cost by more than $w,$ that is,
    $b(best_b) \leq w\, cost^*.$
\end{restatable}

\begin{proof}
Let $\hat{b}_i(N)$ be a budget for agent $a_i$ at a CT node $N$ before the low-level search and $b_i(N)$ be the updated budget for agent $a_i$ after the low-level search. 
By definition, $b(N)=\sum_i^m b_i(N).$
We first prove by induction that $b(N)\leq w\, cost^*(N),$ where $cost^*(N)$ is the optimal solution cost satisfying the constraints in $N.$ 

At the root node $R,$ all $\hat{b}_i$'s are initialized with $\hat{b}_i(R) = w\, \hat{f_i}(R),$ where $\hat{f_i}(R)$ is an admissible heuristic cost estimate of agent $a_i$'s path, that is, $\hat{f_i}(R) \leq f_i^*(R),$ where $f_i^*(R)$ is the shortest path satisfying the constraints which is empty in $R.$
Suppose Y($\hat{b}_i(R)$) finds a path with length $f_i(R).$ 
If $f_i(R) \leq \hat{b}_i(R),$ then $b_i(R)$ is set to $\hat{b}_i(R)=w\, \hat{f_i}(R) \leq w\, f_i^*(R).$
If $f_i(R) > \hat{b}_i(R),$
then $b_i(R)$ is set to $w\,f_i(R) = w\, f_i^*(R),$
where $f_i(R)=f_i^*(R)$ holds by corollary~\ref{corollary:bcoa:necessity}.
In either case, $b_i(R) \leq w\, f_i^*(R).$ 
This holds for all $a_i$'s, and hence
$b(R)=\sum_i^m b_i(R) \leq w\, f_i^*(R) \leq w\, cost^*(R).$
%
Now assume that $b_i(N) \leq w\, f_i^*(N)$ holds for all $a_i$ in a CT node $N$.
We will first show that for a child CT node $N',$ $b_i(N')\leq w\, f_i^*(N')$ also holds for all agents.
Note that for any $i,$ $f_i^*(N) \leq f_i^*(N')$ holds since $N.constraints \subseteq N'.constraints.$
Therefore, for each $i,$ $\hat{b}_i(N') \leq w\, f_i^*(N')$ holds, since $\hat{b}_i(N')$ is set to $b_i(N) \leq w\, f_i^*(N) \leq w\, f_i^*(N').$
Suppose Y($\hat{b}_i(N')$) finds a path with length $f_i(N')$ for agent $a_i$ in $N'$.
If $f_i(N') > \hat{b}_i(N'),$ then 
$b_i(N')$ is updated with $w\,f_i(N')$
so that $b_i(N')= w\, f_i^*(N'),$ 
for $f_i(N')=f_i^*(N')$ by Corollary~\ref{corollary:bcoa:necessity}.
Hence, $b_i(N') \leq w\, f_i^*(N')$ for all $i,$ and $b(N')\leq w\, cost^*(N')$ follows as before.
Therefore, $b(N) \leq w\, cost^*(N)$ for any CT node $N$.
Next, we show that a CT node of CBSB does not contain a constraint that CBS would not generate.

Let $T$ be a constraint on agent $a_i$ in a CT node $N'$ and 
let $N$ be a precedent node of $N',$ where the conflict $F$ that resulted in the constraint $T$ was first found. 
Note that $F$ was found by, say, agent $a_j$ with Y($\hat{b}_j(N)$).
Now, suppose Y(0) was used instead in $N$ for agent $a_j.$
Since $\hat{b}_j(N)>0,$ Y(0) would not find fewer conflicts than Y($B_j(N)$) according to Corollary~\ref{corollary:bcoa:budget_conflict}.
Hence, Y(0) would also result in the conflict $F$ in $N$. 
But Y(0) is equivalent to A*, so A* would have found the conflict $F$ in $N$ for the agent $a_j.$
Therefore, at least one CT node $N$ exists in OPEN whose $cost^*(N)\leq cost^*.$
But $b(N)\leq w\, cost^*(N),$ so the minimum $b$ in OPEN does not overestimate the true optimal solution $cost^*$ by more than $w.$
\end{proof}

\begin{restatable}{theorem}{CBSBBoundedSuboptimality}   	\label{theorem:cbsy:boundedsuboptimality}
	CBSB returns a solution whose cost does not exceed the optimal solution cost by more than $w$. 
\end{restatable}

\begin{proof}
    CBSB expands a CT node whose $cost$ does not exceed the minimum $b$ in OPEN. 
    Hence, when a goal CT node is expanded from CBSB, 
    the cost does not exceed the optimal solution cost by more than $w$ by Theorem~\ref{theorem:cbsy:minimumB}.
\end{proof}

\begin{figure*}[ht]
	\centering
	\def\figuresize{0.33}
	\begin{subfigure}{1\textwidth}
        \begin{tikzpicture}
        \node(a){\includegraphics[width=1\linewidth,trim={17mm 30mm 10mm 9mm},clip]{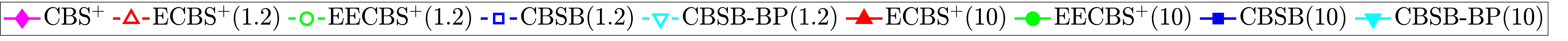}};
        \end{tikzpicture}
	\end{subfigure}\\[-1.5ex]	
	\begin{subfigure}{\figuresize\textwidth}
		\begin{tikzpicture}
			\node(a){\includegraphics[width=\myLineScale\linewidth,trim={2mm 1mm 8mm 0mm},clip]{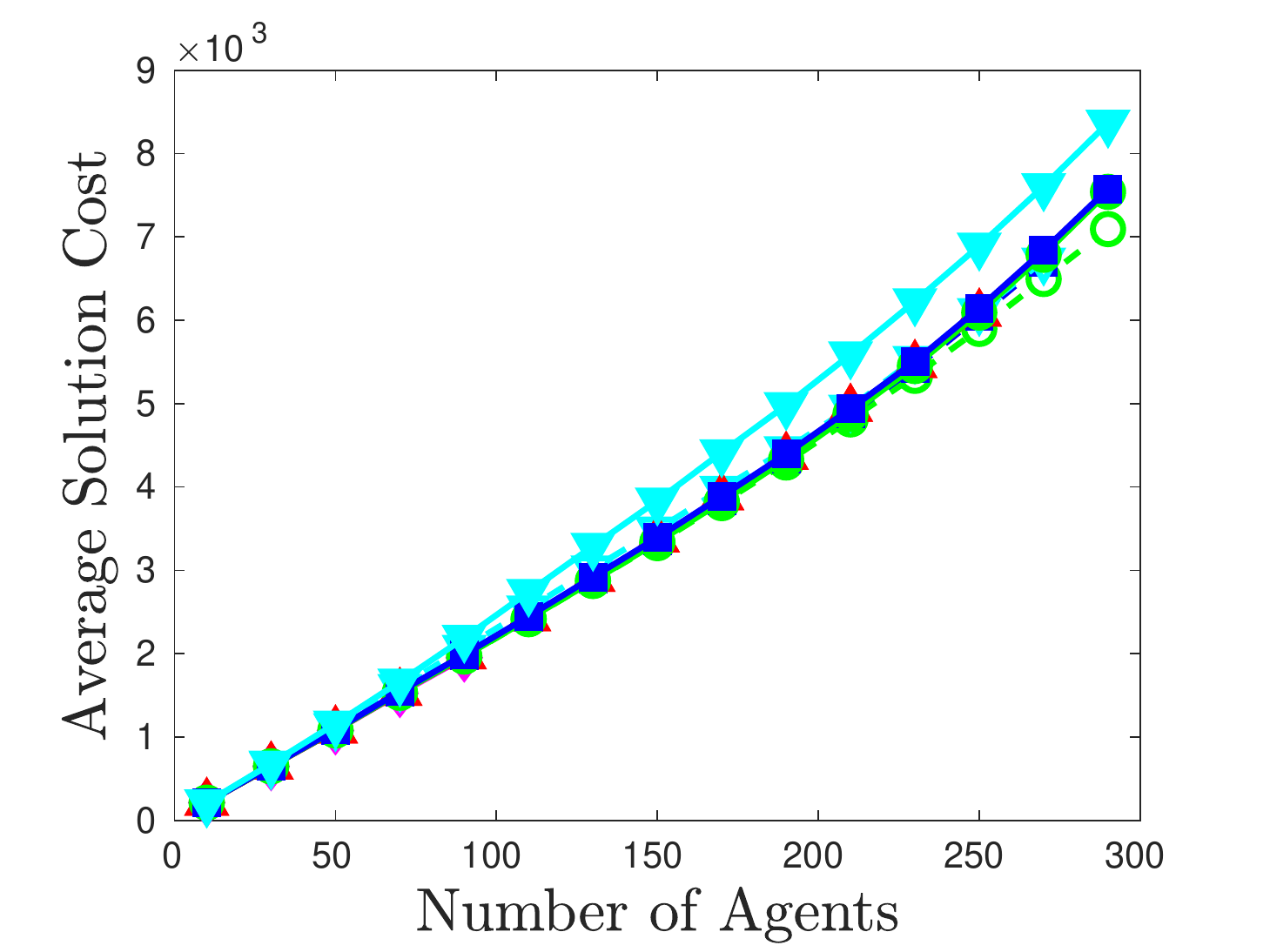}};
			\node at (a.north west)
			[
			anchor=center,
			xshift=48.5mm,
			yshift=-33.5mm
			]
			{
				\includegraphics[width=0.3\linewidth]{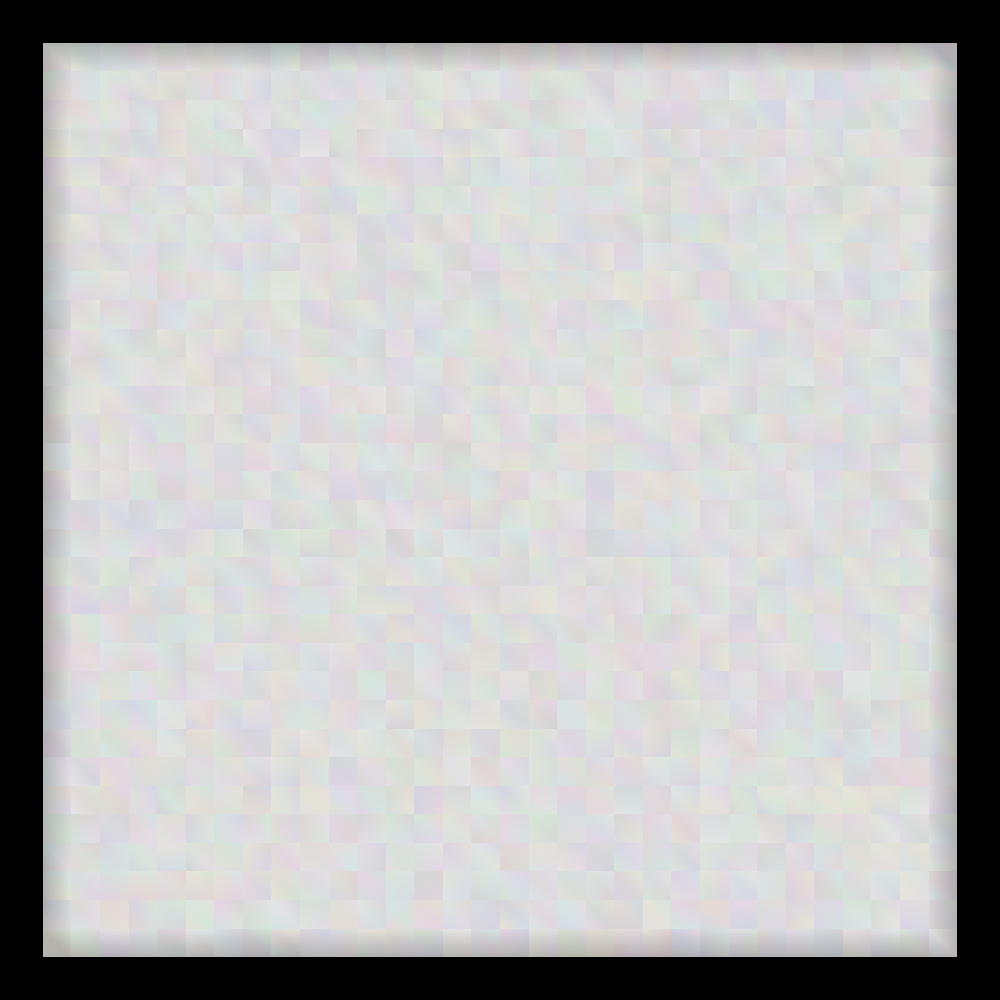}
			};
			\node at (a.north west)
			[
			anchor=center,
			xshift=48.5mm,
			yshift=-23mm
			]
			{
				\tiny{empty-32-32}
			};
		\end{tikzpicture}
	\end{subfigure}\hfill
	\begin{subfigure}{\figuresize\textwidth}
		\begin{tikzpicture}
			\node(a){\includegraphics[width=\myLineScale\linewidth,trim={2mm 1mm 8mm 0mm},clip]{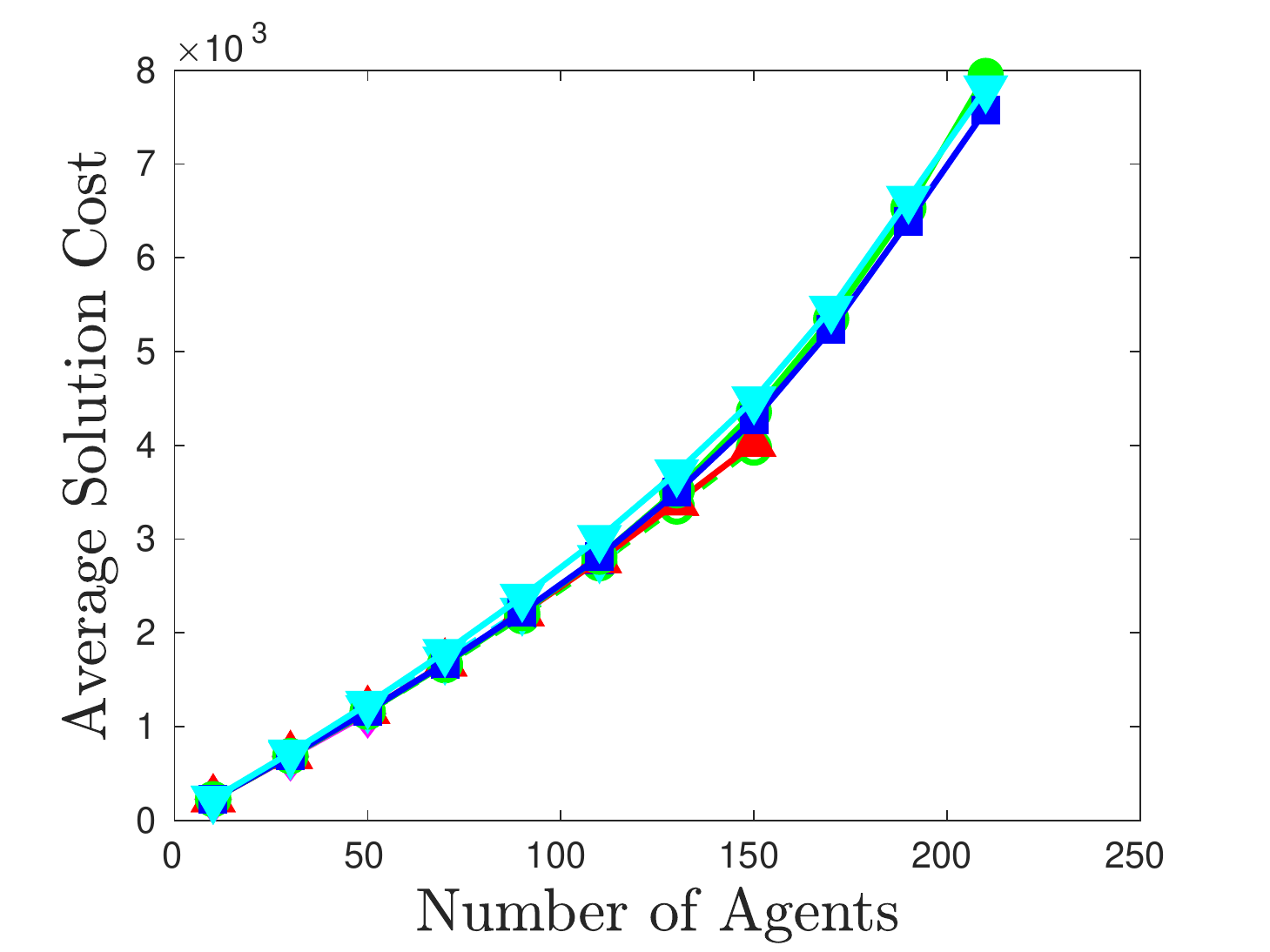}};
			\node at (a.north west)
			[
			anchor=center,
			xshift=48.5mm,
			yshift=-33.5mm
			]
			{
				\includegraphics[width=0.3\linewidth]{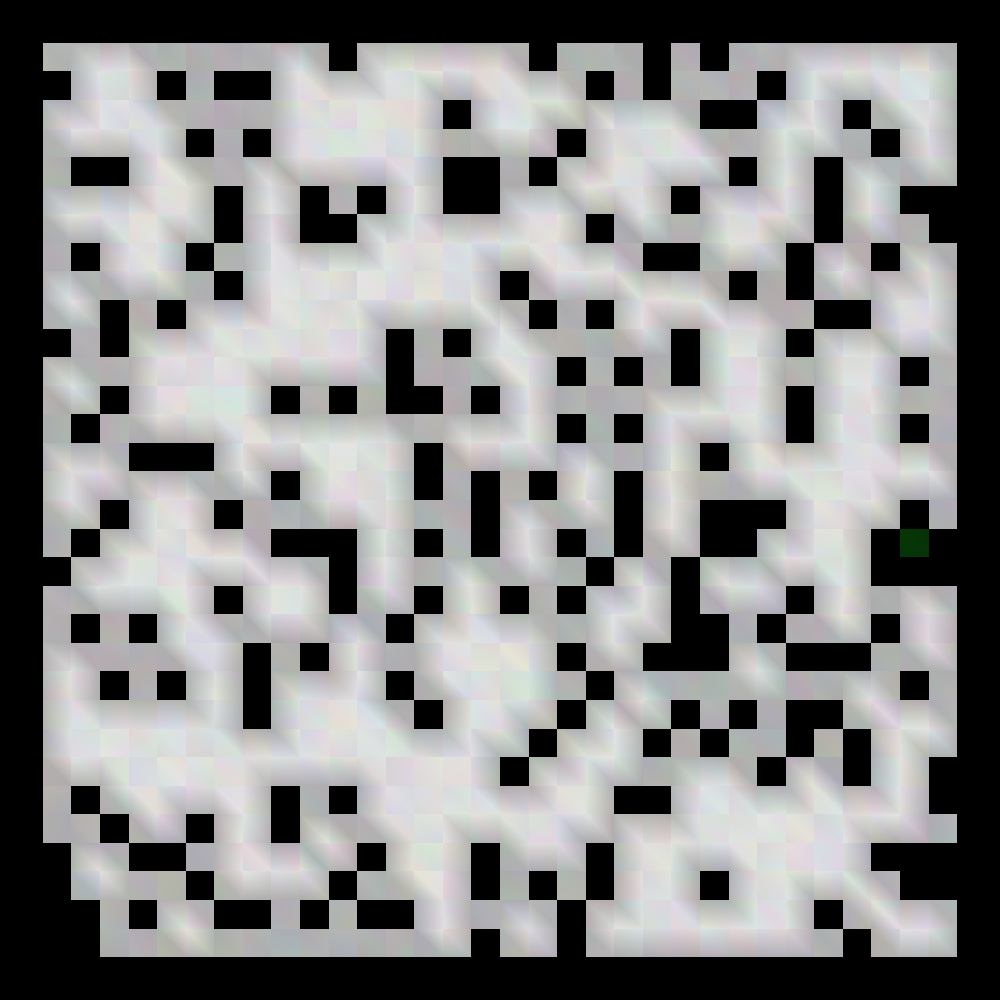}
			};
			\node at (a.north west)
			[
			anchor=center,
			xshift=48.5mm,
			yshift=-23mm
			]
			{
				\tiny{radom-32-32-20}
			};
		\end{tikzpicture}
	\end{subfigure}
	\begin{subfigure}{\figuresize\textwidth}
		\begin{tikzpicture}
			\node(a){\includegraphics[width=\myLineScale\linewidth,trim={0mm 1mm 8mm 0mm},clip]{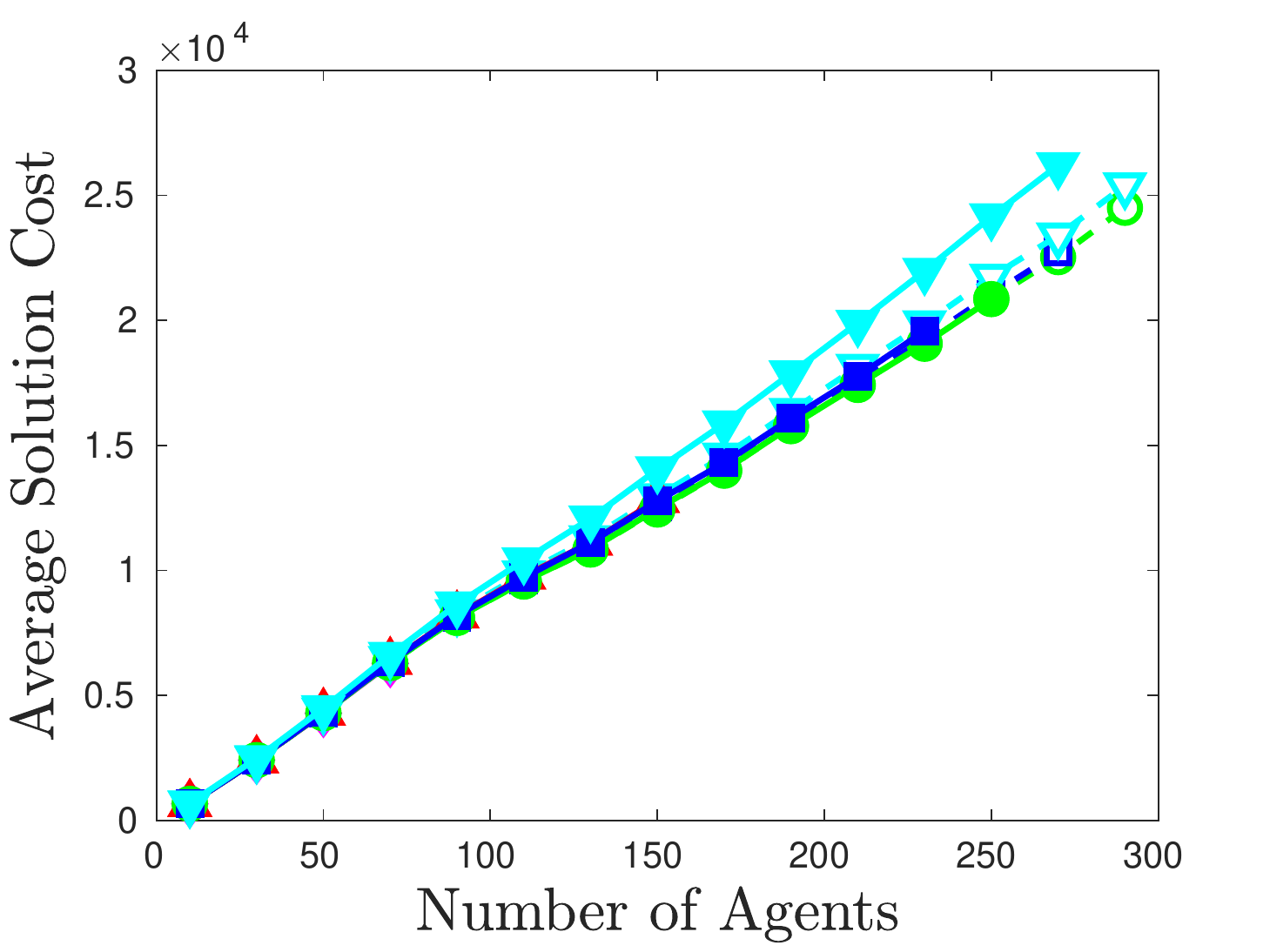}};
			\node at (a.north west)
			[
			anchor=center,
			xshift=48.5mm,
			yshift=-33mm
			]
			{
				\includegraphics[width=0.3\linewidth]{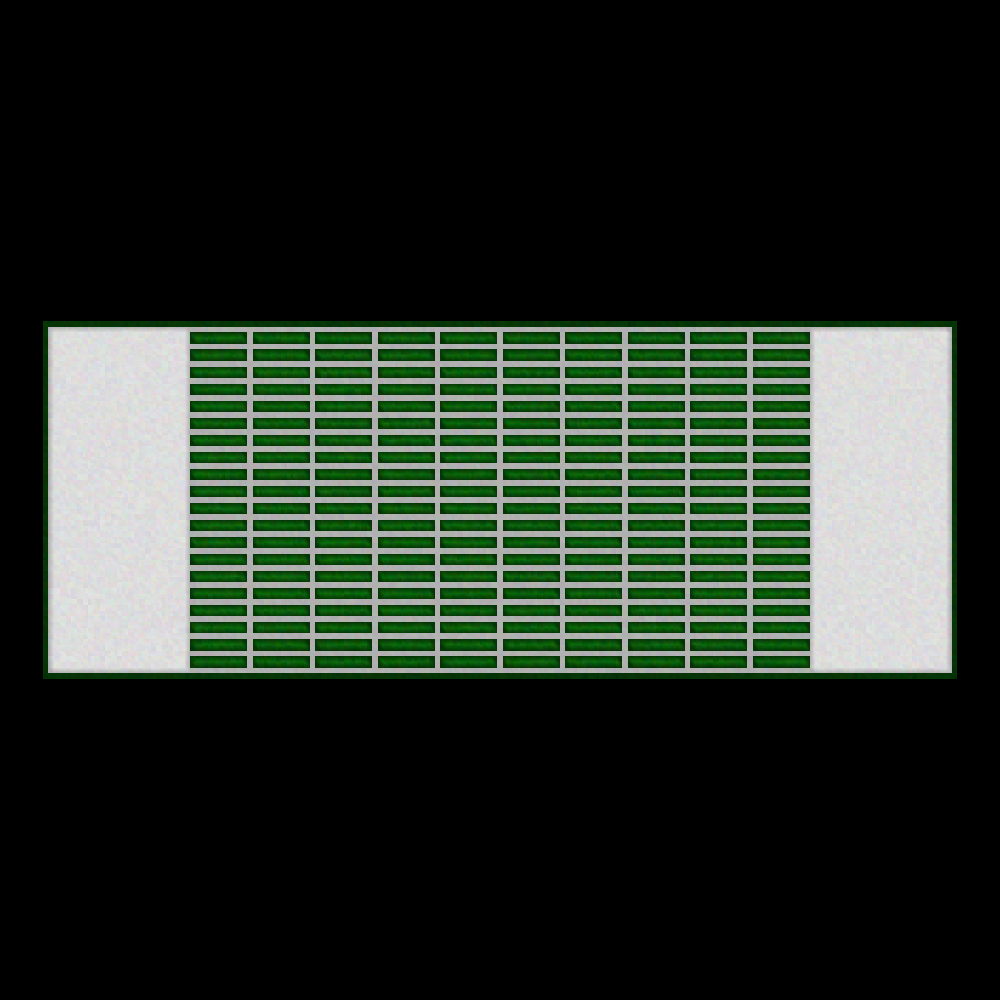}
			};
			\node at (a.north west)
			[
			anchor=center,
			xshift=48.5mm,
			yshift=-23mm
			]
			{
				\tiny{warehouse-10-20-10-2-1}
			};
		\end{tikzpicture}
	\end{subfigure}\\[-1.5ex]
	\begin{subfigure}{\figuresize\textwidth}
		\begin{tikzpicture}
			\node(a){\includegraphics[width=\myLineScale\linewidth,trim={2mm 1mm 8mm 0mm},clip]{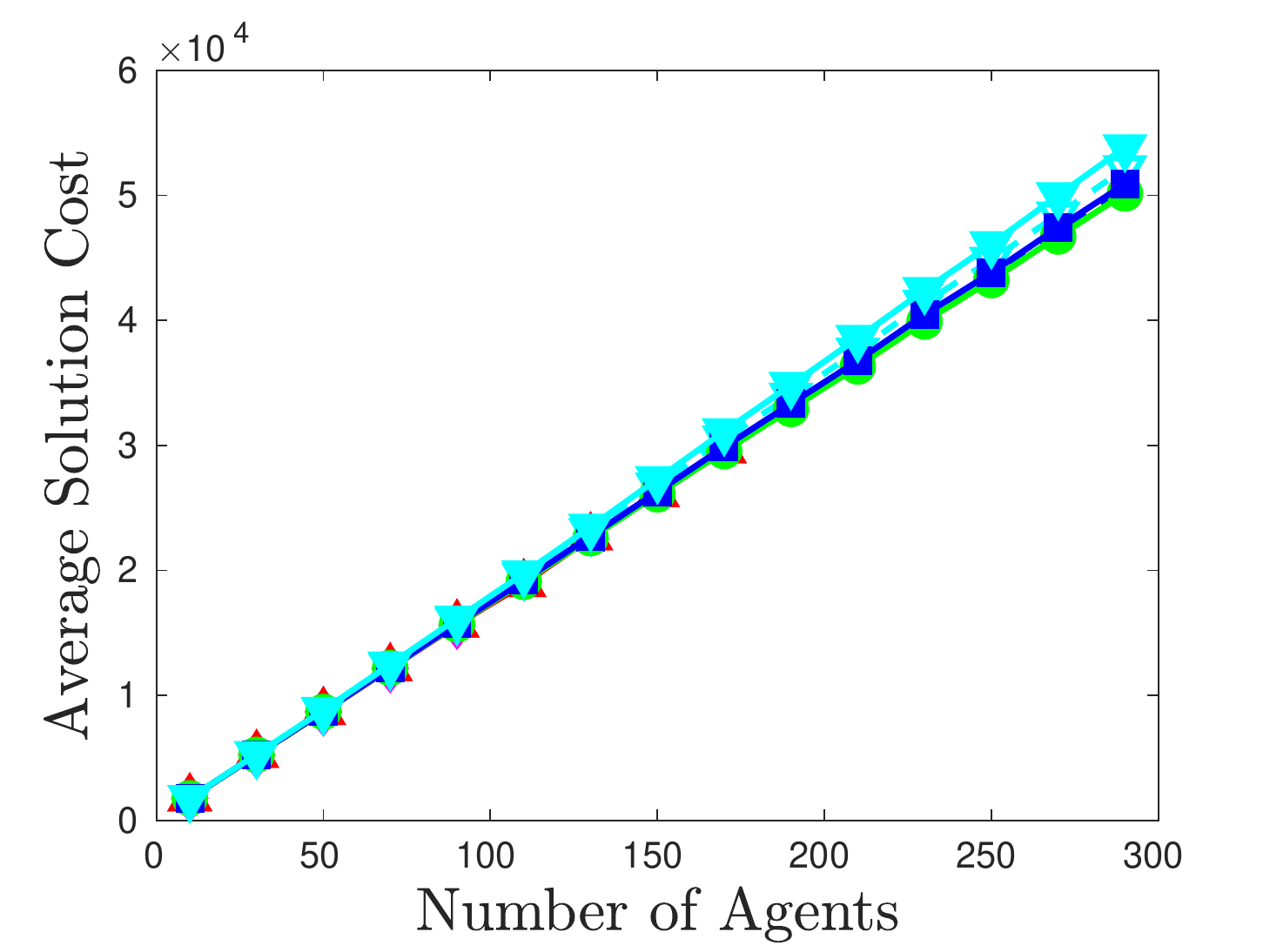}};
			\node at (a.north west)
			[
			anchor=center,
			xshift=48.5mm,
			yshift=-33.5mm
			]
			{
				\includegraphics[width=0.3\linewidth]{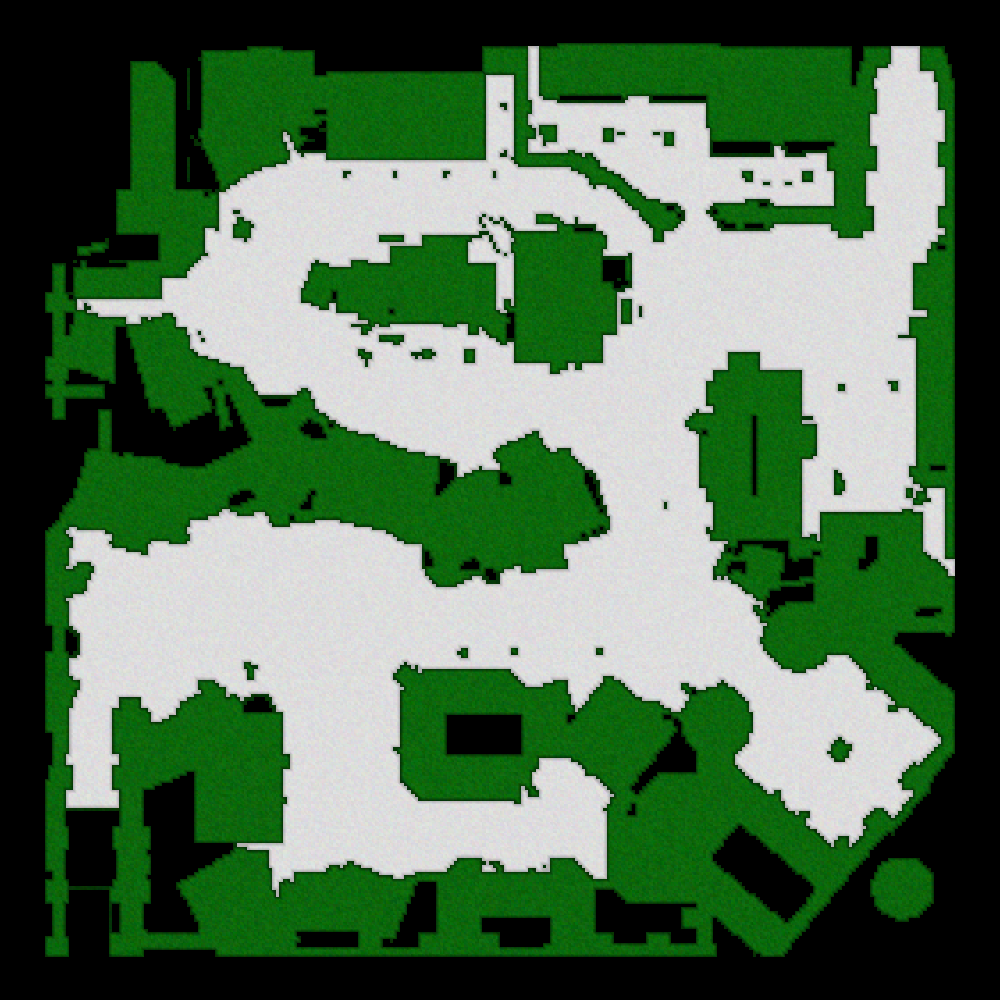}
			};
			\node at (a.north west)
			[
			anchor=center,
			xshift=48.5mm,
			yshift=-23mm
			]
			{
				\tiny{den520d}
			};
		\end{tikzpicture}
	\end{subfigure}
	\begin{subfigure}{\figuresize\textwidth}
		\begin{tikzpicture}
			\node(a){\includegraphics[width=\myLineScale\linewidth,trim={2mm 1mm 8mm 0mm},clip]{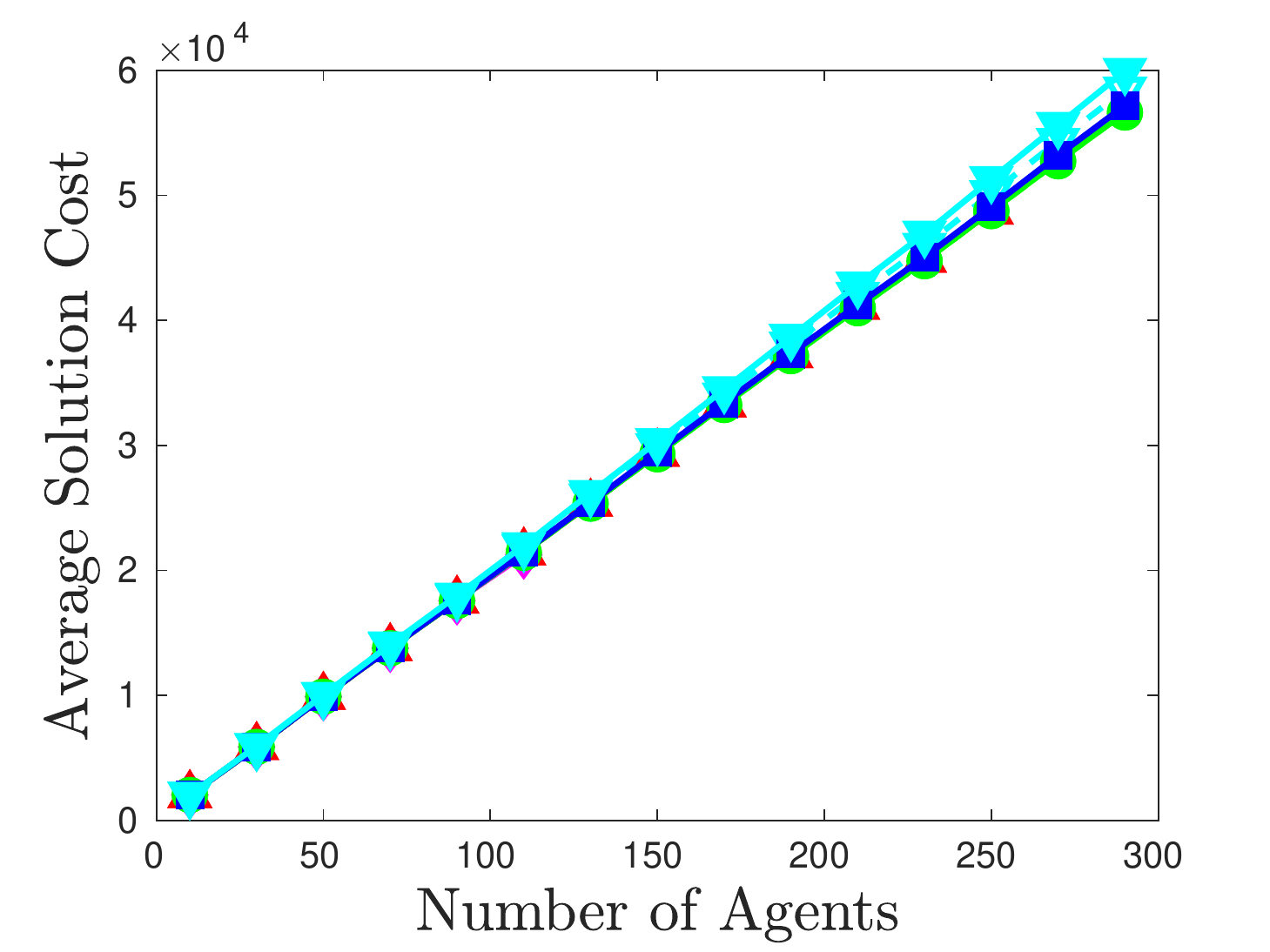}};
			\node at (a.north west)
			[
			anchor=center,
			xshift=48.5mm,
			yshift=-33.5mm
			]
			{
				\includegraphics[width=0.3\linewidth]{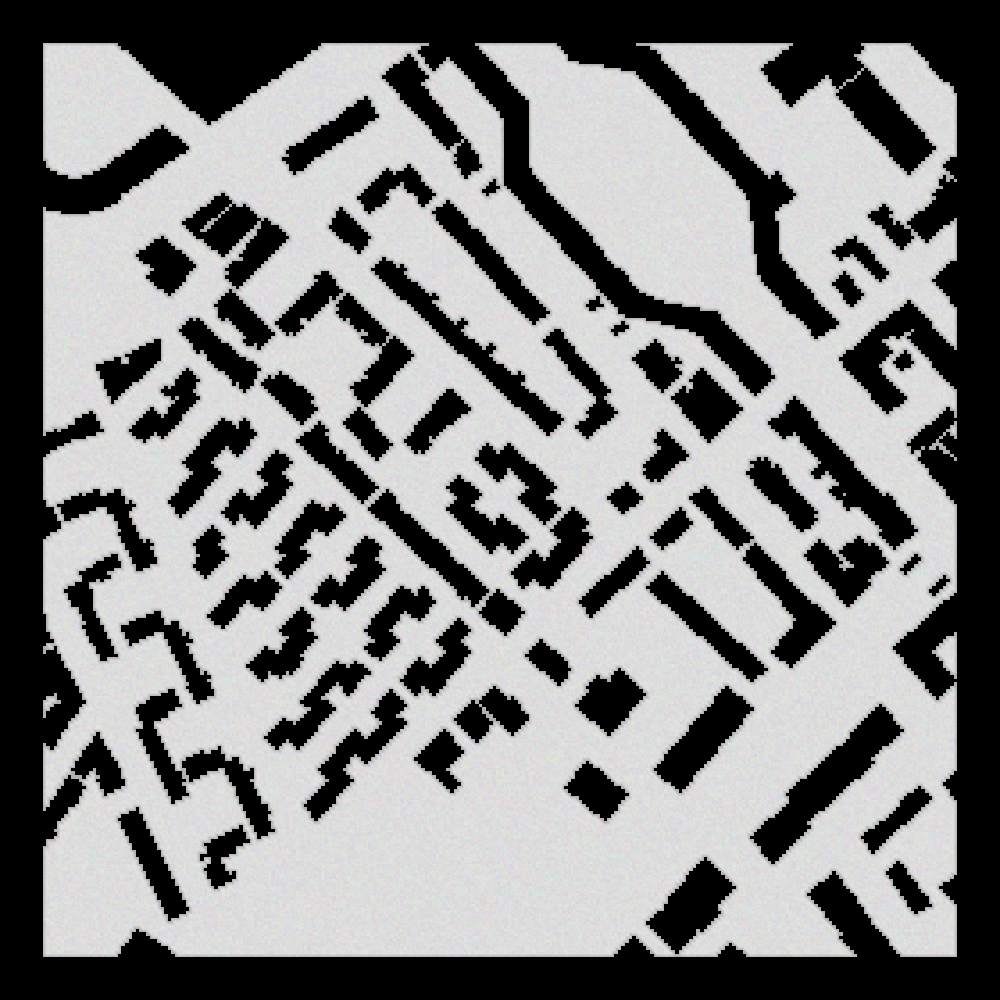}
			};
			\node at (a.north west)
			[
			anchor=center,
			xshift=48.5mm,
			yshift=-23mm
			]
			{
				\tiny{Boston 0 256}
			};
		\end{tikzpicture}
	\end{subfigure}\hfill
	\begin{subfigure}{\figuresize\textwidth}
		\begin{tikzpicture}
			\node(a){\includegraphics[width=\myLineScale\linewidth,trim={2mm 1mm 8mm 0mm},clip]{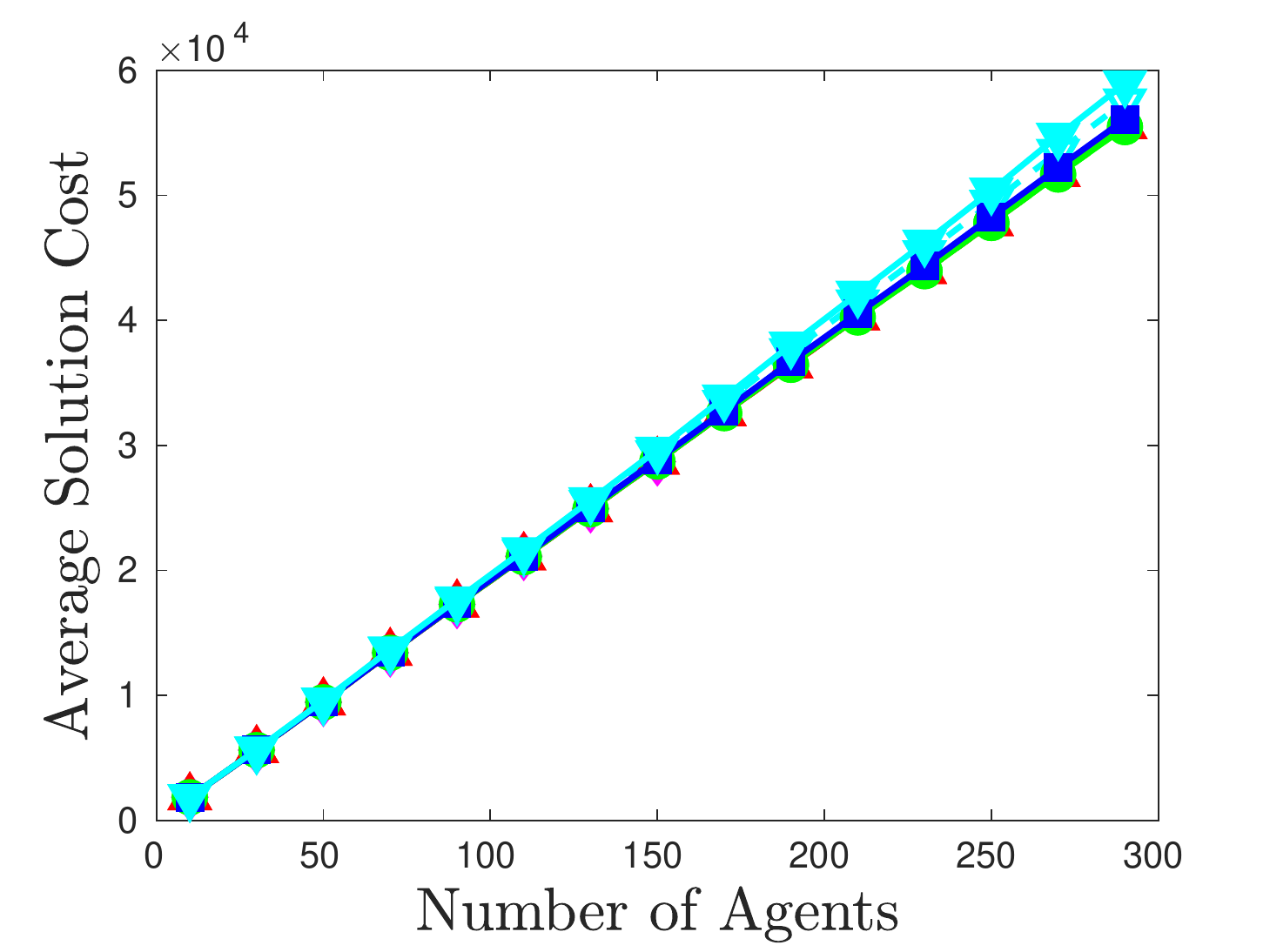}};
			\node at (a.north west)
			[
			anchor=center,
			xshift=48.5mm,
			yshift=-33.5mm
			]
			{
				\includegraphics[width=0.3\linewidth]{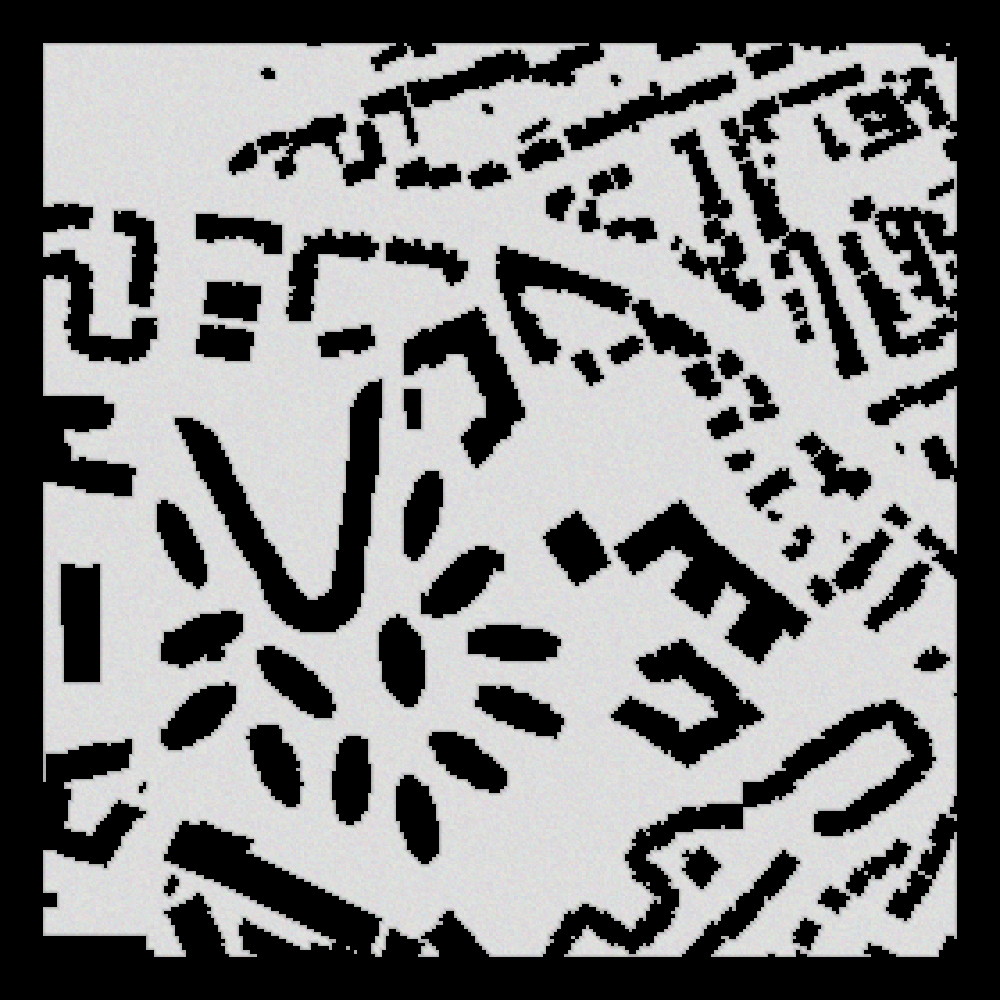}
			};
			\node at (a.north west)
			[
			anchor=center,
			xshift=48.5mm,
			yshift=-23mm
			]
			{
				\tiny{Paris 1 256}
			};
		\end{tikzpicture}
	\end{subfigure}\hfill
	\caption{Cost of solution found by CBS$^+$, ECBS$^+$, EECBS$^+$, and CBSB(-BP), averaged over 25 random scenarios with different number of agents in 6 different maps. The suboptiamlity factor for ECBS$^+$, EECBS$^+$, and CBSB are set to 1.2 and 10. In all the instances, ECBS$^+$, EECBS$^+$, and CBSB finds a near-optimal solution. 
	A missing marker indicates that the corresponding solver could not find a solution in more than 75\% of the 25 random instances given a time limit of 10 seconds.
    }
	\label{cbsb:f:cost_comparison}
\end{figure*}
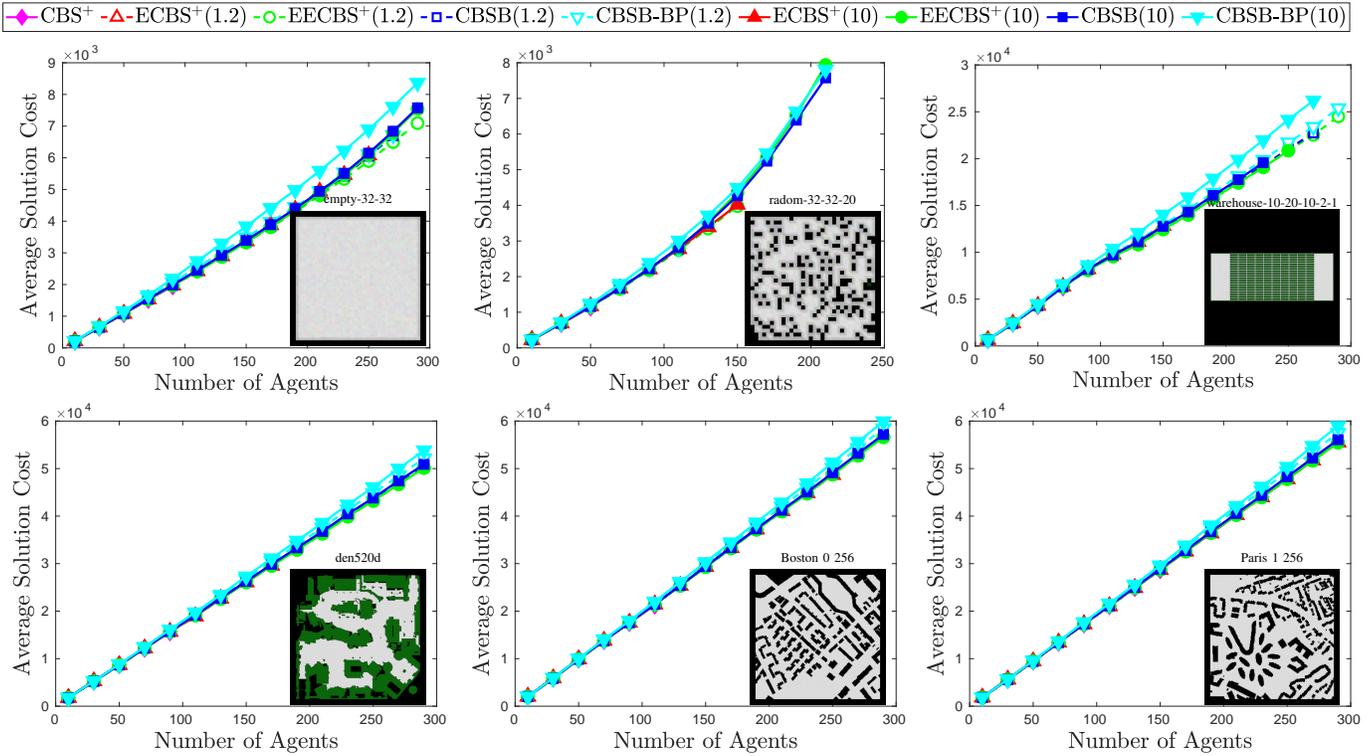

\begin{figure*}[ht]
	\centering
	\begin{subfigure}{1\textwidth}
        \begin{tikzpicture}
        \node(a){\includegraphics[width=\myLineScale\linewidth,trim={10mm 30mm 10mm 9mm},clip]{figures/legend2.pdf}};
        \end{tikzpicture}
	\end{subfigure}\\[-1ex]	
	\begin{subfigure}{\myMSFigureScale\textwidth}
        \begin{tikzpicture}
        \node(a){\includegraphics[width=\myLineScale\linewidth,trim={1mm 0 2mm 6mm},clip]{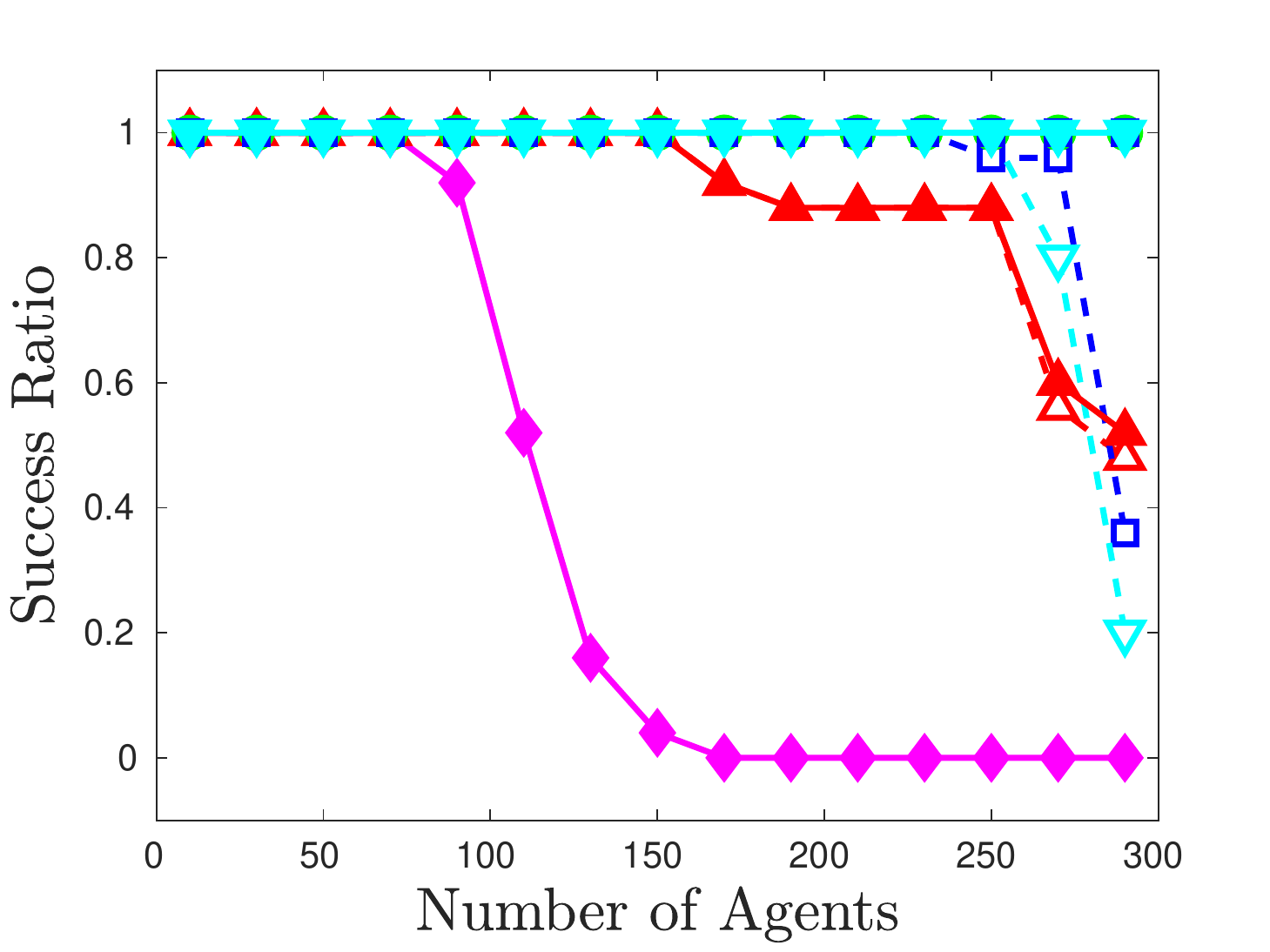}};
        \end{tikzpicture}
	\end{subfigure}\hfill
	\begin{subfigure}{\myMSFigureScale\textwidth}
        \begin{tikzpicture}
        \node(a){\includegraphics[width=\myLineScale\linewidth,trim={1mm 0 2mm 6mm},clip]{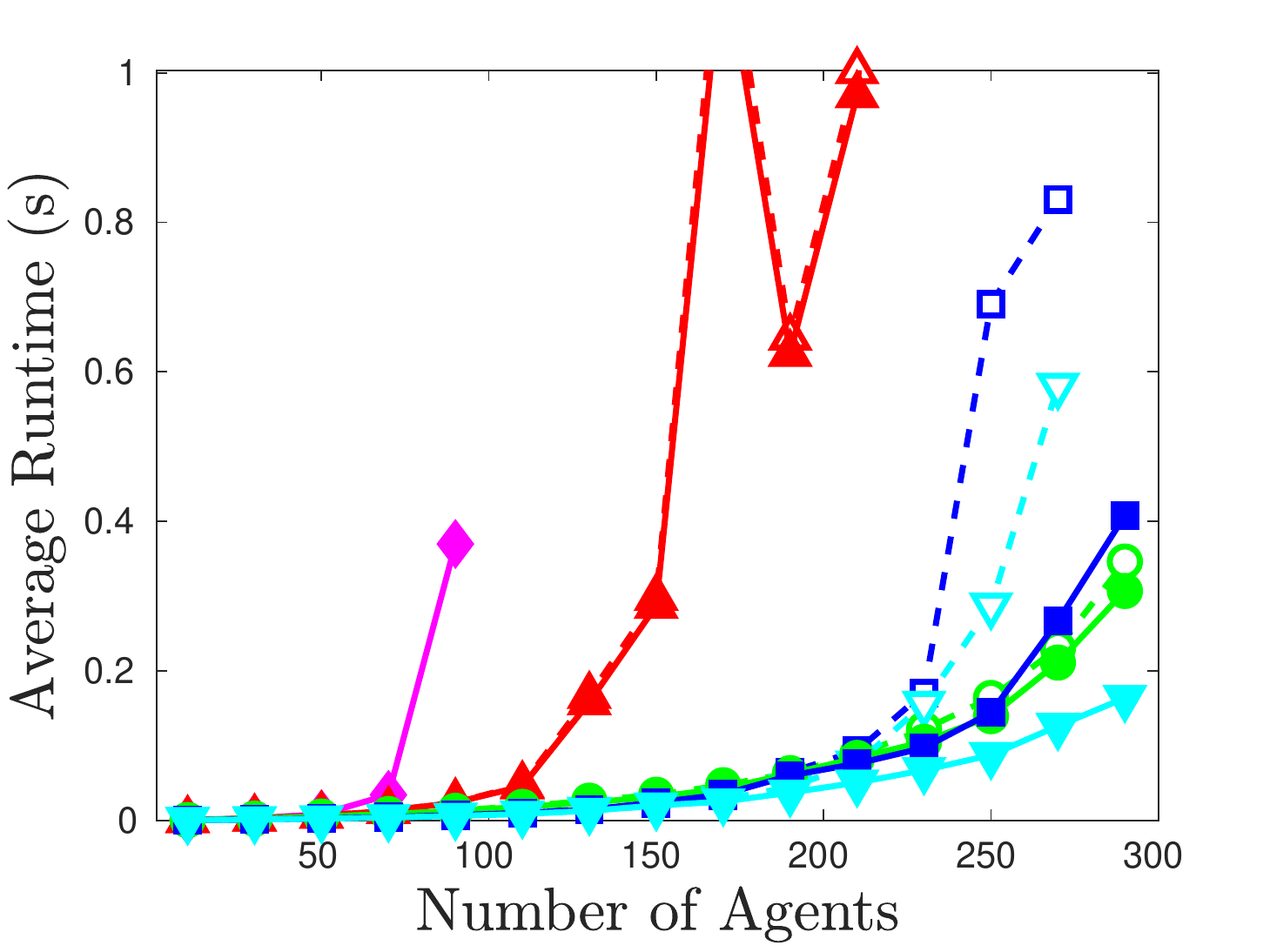}};
        \node at (a.north west)
        [
        anchor=center,
        xshift=12mm,
        yshift=-8mm
        ]
        {
            \includegraphics[width=0.3\linewidth]{figures/empty}
        };
         \node at (a.north west)
	    [
	    anchor=center,
	    xshift=12mm,
	    yshift=-16mm
	    ]
	    {
	    	\tiny{empty-32-32}
	    };
        \end{tikzpicture}
	\end{subfigure}\hfill
	\begin{subfigure}{\myMSFigureScale\textwidth}
        \begin{tikzpicture}
        \node(a){\includegraphics[width=\myLineScale\linewidth,trim={1mm 0 2mm 6mm},clip]{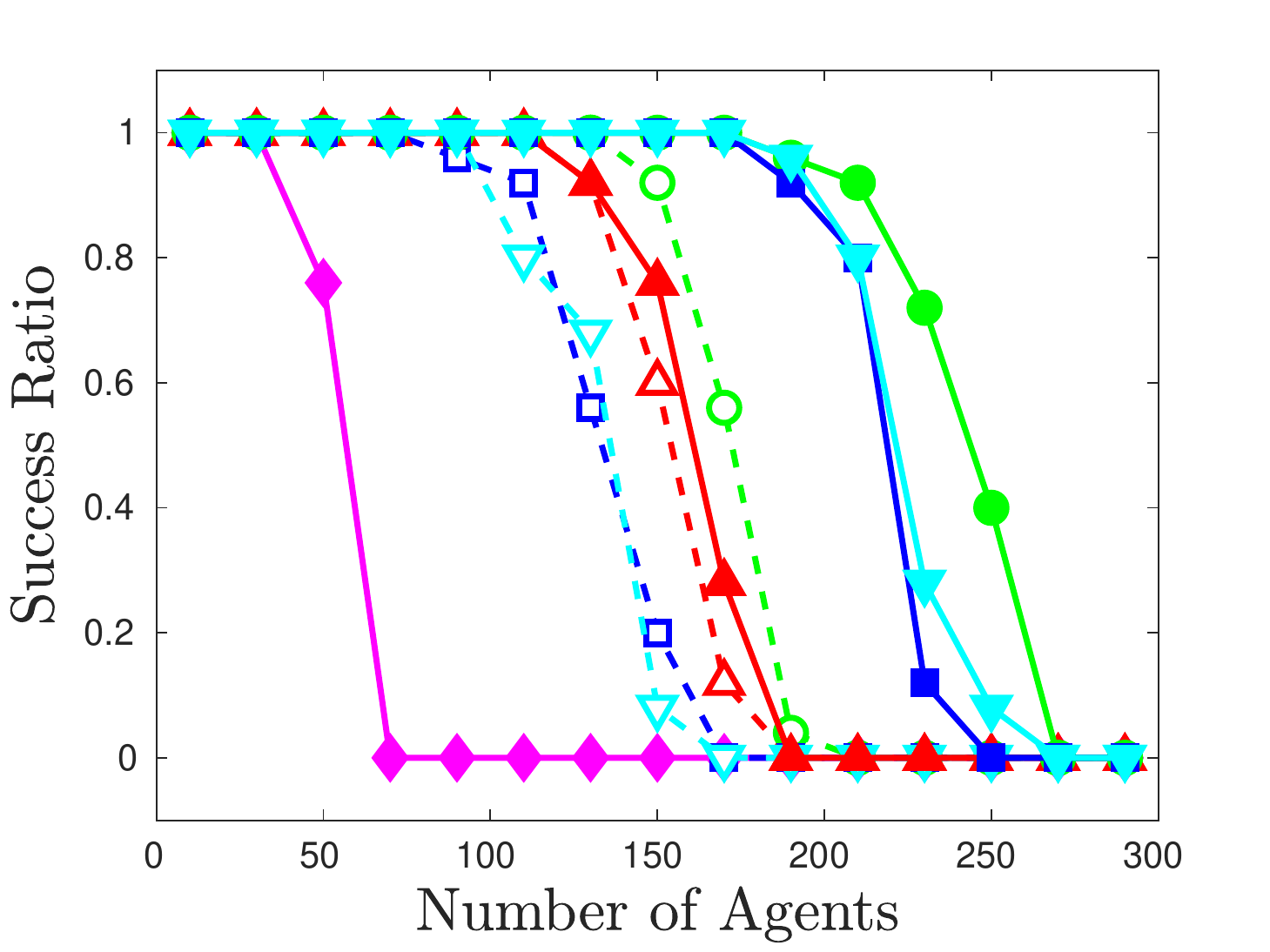}};
        \end{tikzpicture}
	\end{subfigure}\hfill
	\begin{subfigure}{\myMSFigureScale\textwidth}
        \begin{tikzpicture}
        \node(a){\includegraphics[width=\myLineScale\linewidth,trim={1mm 0 2mm 6mm},clip]{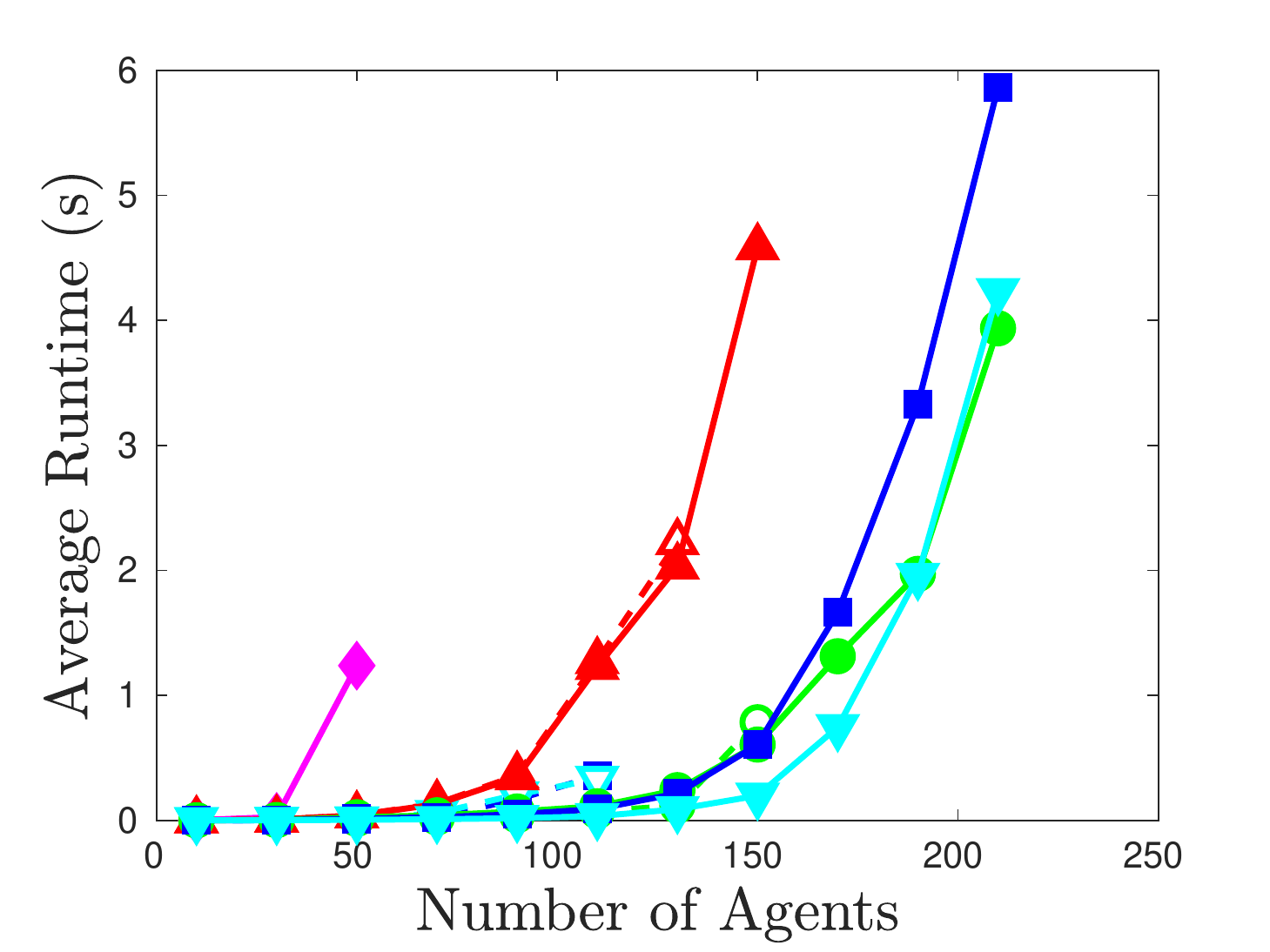}};
                \node at (a.north west)
        [
        anchor=center,
        xshift=12mm,
        yshift=-8mm
        ]
        {
            \includegraphics[width=0.3\linewidth]{figures/random}
        };
             \node at (a.north west)
	    [
	    anchor=center,
	    xshift=12mm,
	    yshift=-16mm
	    ]
	    {
	    	\tiny{random-32-32-20}
	    };
        \end{tikzpicture}
	\end{subfigure}\\[-1ex]
	\begin{subfigure}{\myMSFigureScale\textwidth}
        \begin{tikzpicture}
        \node(a){\includegraphics[width=\myLineScale\linewidth,trim={1mm 0 2mm 6mm},clip]{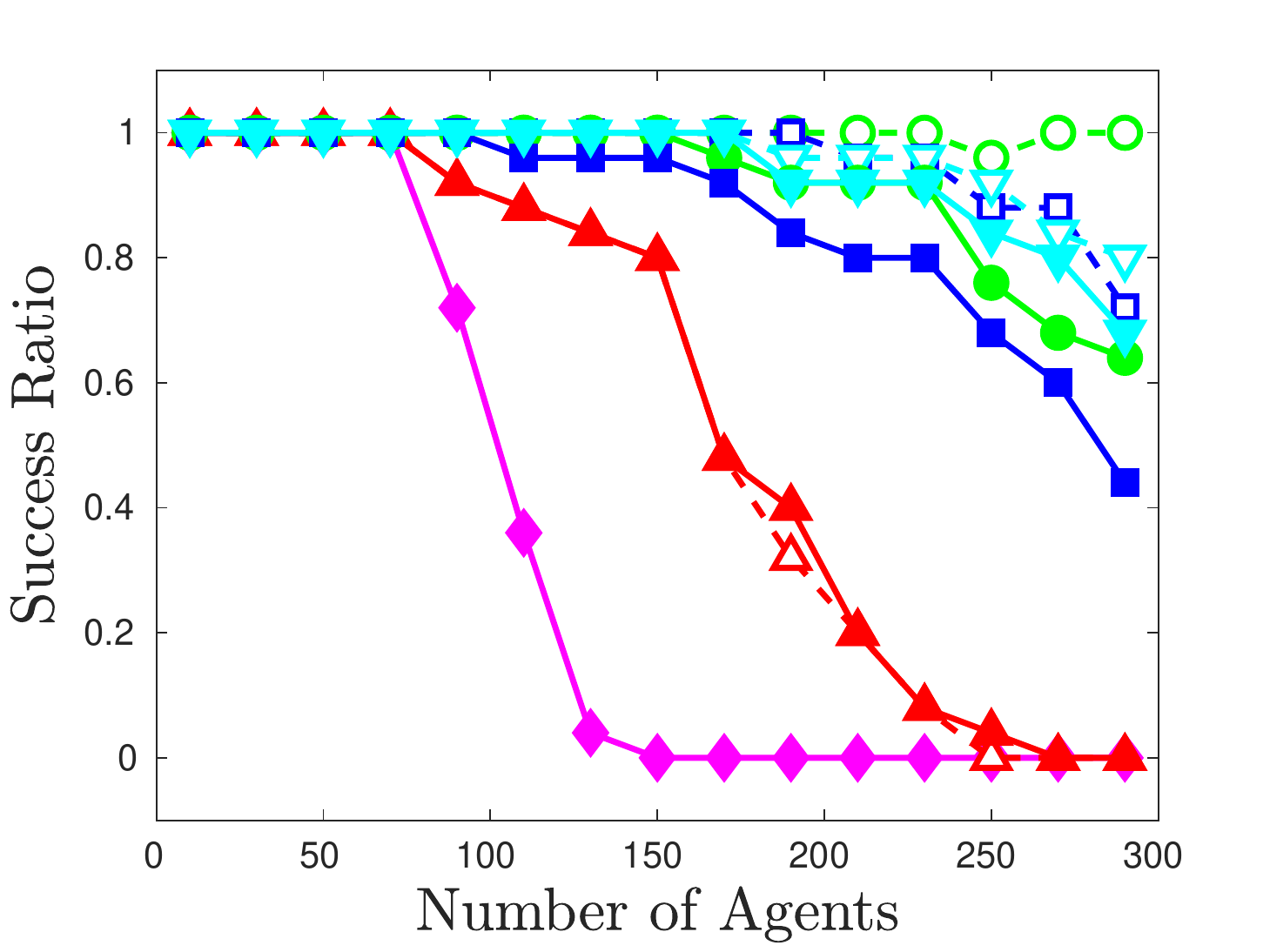}};
        \end{tikzpicture}
	\end{subfigure}\hfill
	\begin{subfigure}{\myMSFigureScale\textwidth}
        \begin{tikzpicture}
        \node(a){\includegraphics[width=\myLineScale\linewidth,trim={1mm 0 2mm 6mm},clip]{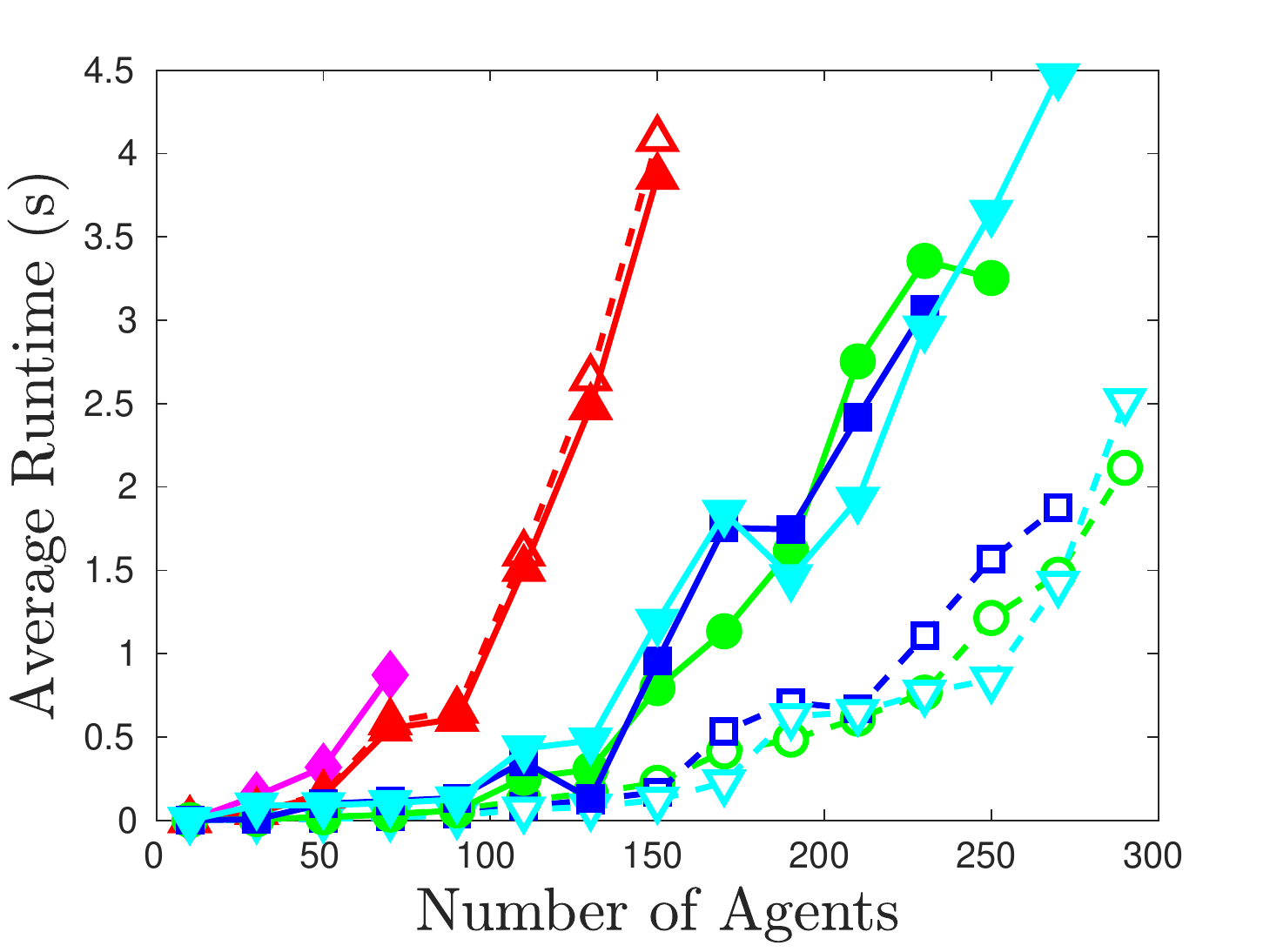}};
                \node at (a.north west)
        [
        anchor=center,
        xshift=12mm,
        yshift=-8mm
        ]
        {
            \includegraphics[width=0.3\linewidth]{figures/warehouse}
        };
             \node at (a.north west)
	    [
	    anchor=center,
	    xshift=12mm,
	    yshift=-16mm
	    ]
	    {
	    	\tiny{warehouse-10-20-10-2-1}
	    };
        \end{tikzpicture}
	\end{subfigure}\hfill
	\begin{subfigure}{\myMSFigureScale\textwidth}
        \begin{tikzpicture}
        \node(a){\includegraphics[width=\myLineScale\linewidth,trim={1mm 0 2mm 6mm},clip]{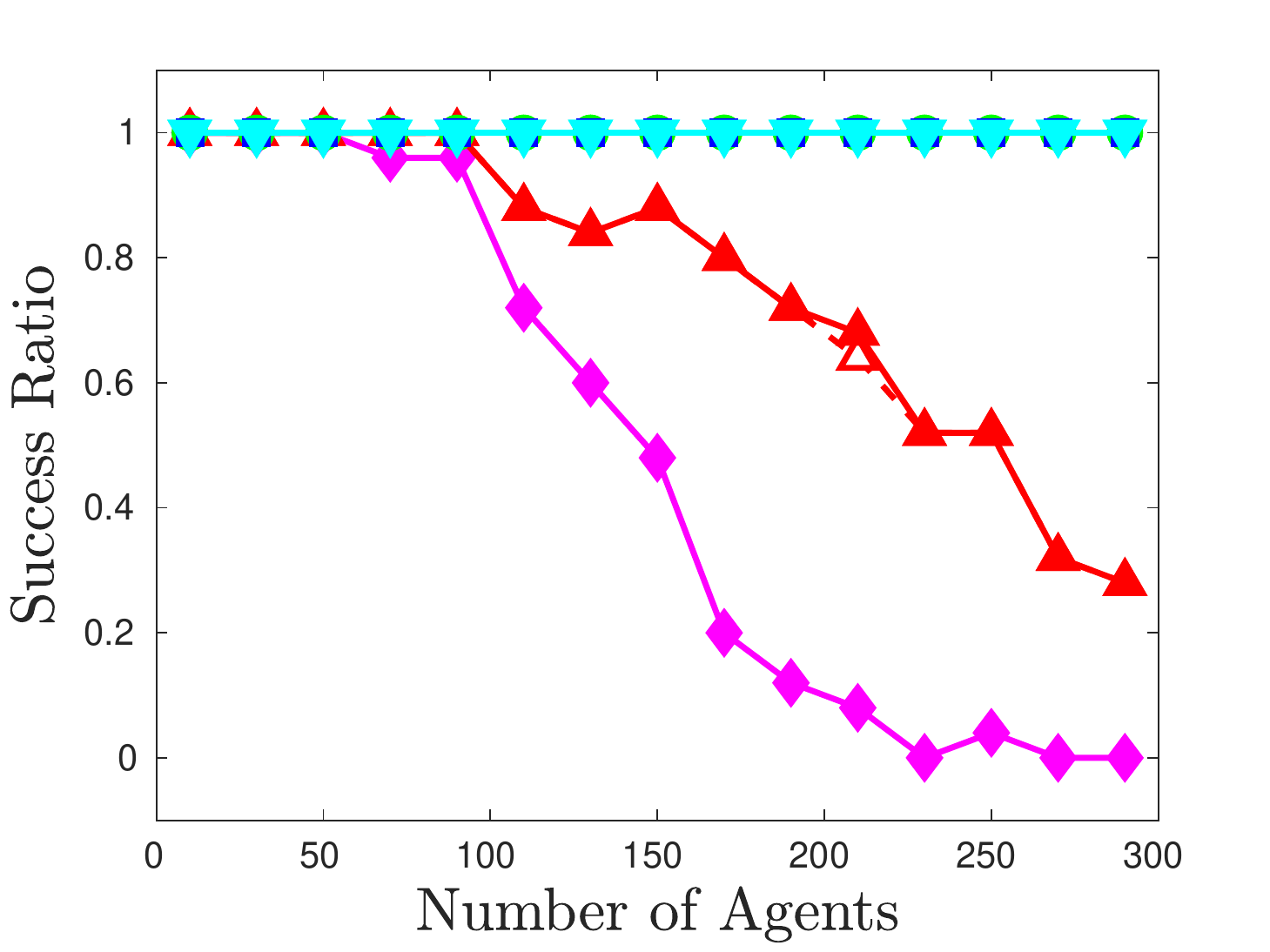}};
        \end{tikzpicture}
	\end{subfigure}\hfill
	\begin{subfigure}{\myMSFigureScale\textwidth}
        \begin{tikzpicture}
        \node(a){\includegraphics[width=\myLineScale\linewidth,trim={1mm 0 2mm 6mm},clip]{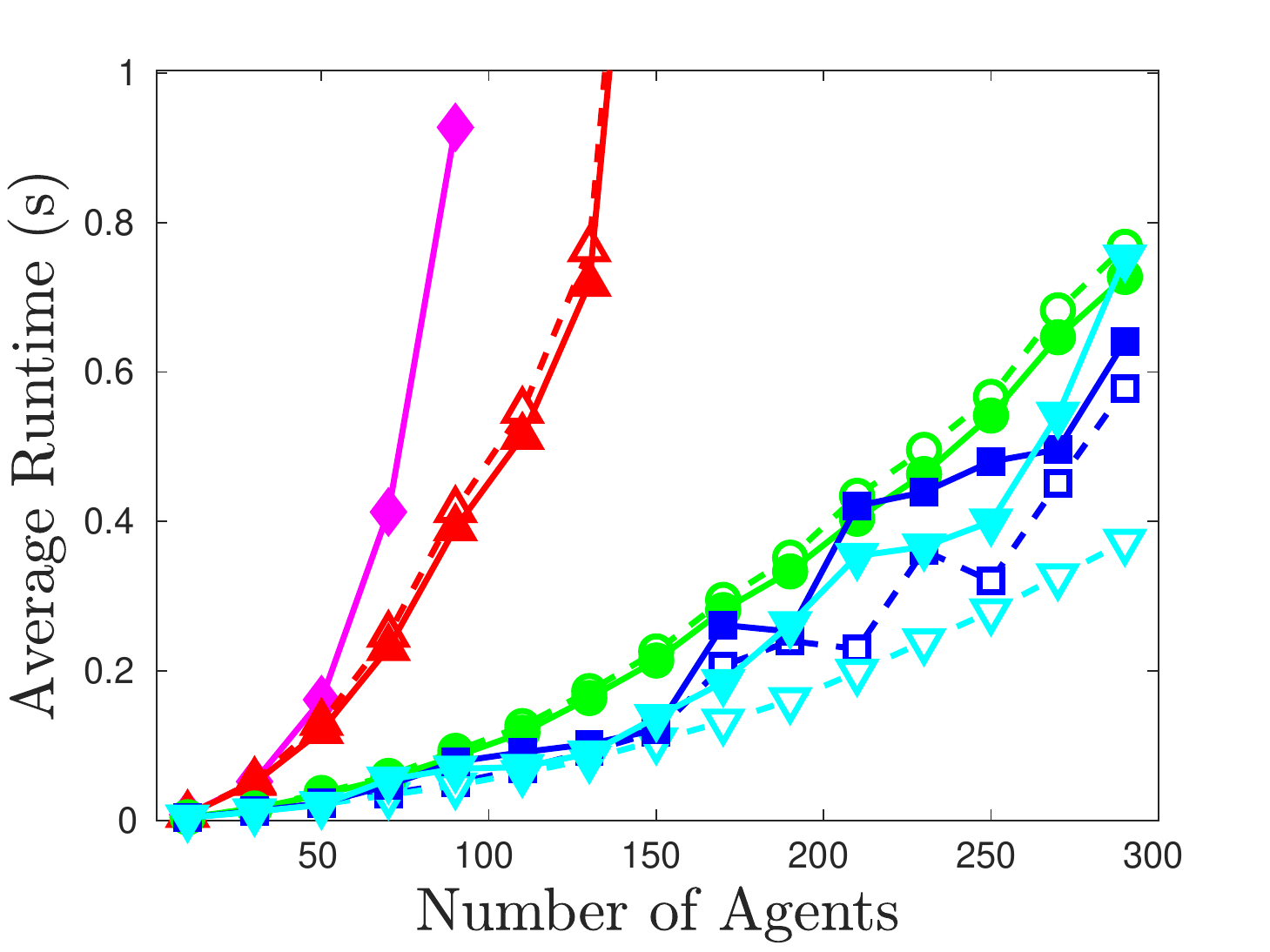}};
                \node at (a.north west)
        [
        anchor=center,
        xshift=12mm,
        yshift=-8mm
        ]
        {
            \includegraphics[width=0.3\linewidth]{figures/den}
        };
             \node at (a.north west)
	    [
	    anchor=center,
	    xshift=12mm,
	    yshift=-16mm
	    ]
	    {
	    	\tiny{den520d}
	    };
        \end{tikzpicture}
	\end{subfigure}\\[-1ex]
	\begin{subfigure}{\myMSFigureScale\textwidth}
        \begin{tikzpicture}
        \node(a){\includegraphics[width=\myLineScale\linewidth,trim={1mm 0 2mm 6mm},clip]{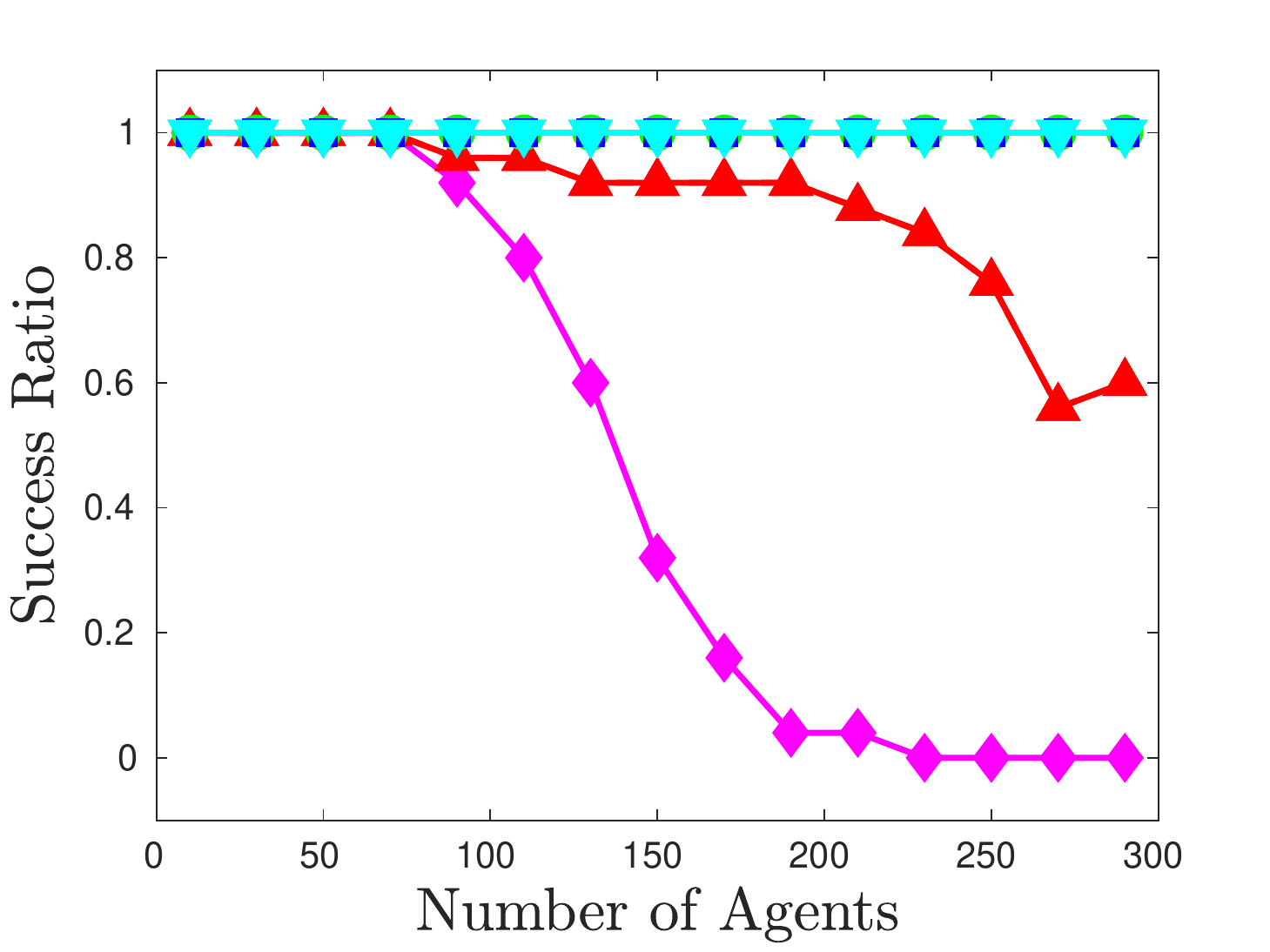}};
        \end{tikzpicture}
	\end{subfigure}\hfill
	\begin{subfigure}{\myMSFigureScale\textwidth}
        \begin{tikzpicture}
        \node(a){\includegraphics[width=\myLineScale\linewidth,trim={1mm 0 2mm 6mm},clip]{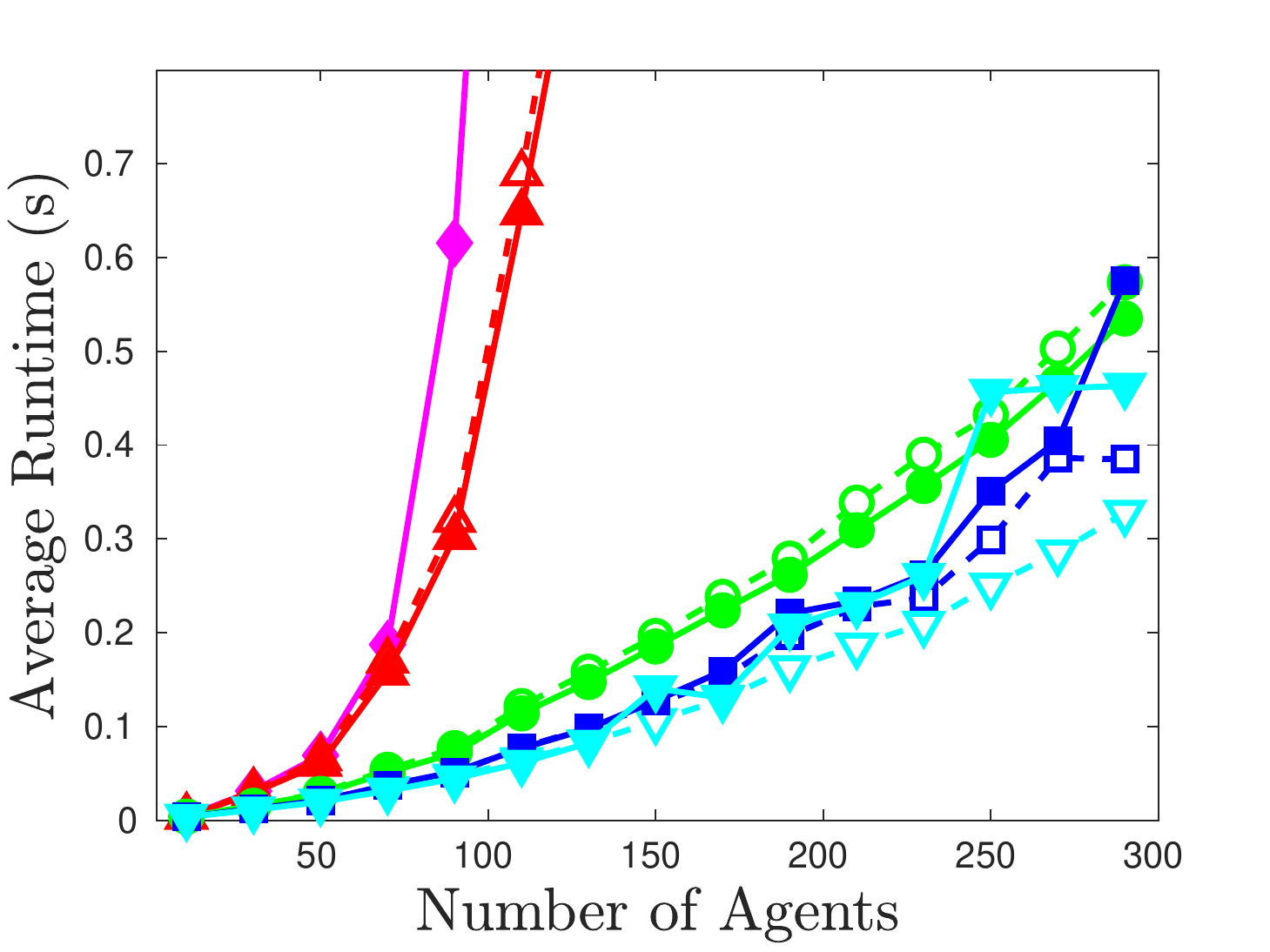}};
                \node at (a.north west)
        [
        anchor=center,
        xshift=12mm,
        yshift=-8mm
        ]
        {
            \includegraphics[width=0.3\linewidth]{figures/boston}
        };
             \node at (a.north west)
	    [
	    anchor=center,
	    xshift=12mm,
	    yshift=-16mm
	    ]
	    {
	    	\tiny{Boston 0 256}
	    };
        \end{tikzpicture}
	\end{subfigure}\hfill
	\begin{subfigure}{\myMSFigureScale\textwidth}
        \begin{tikzpicture}
        \node(a){\includegraphics[width=\myLineScale\linewidth,trim={1mm 0 2mm 6mm},clip]{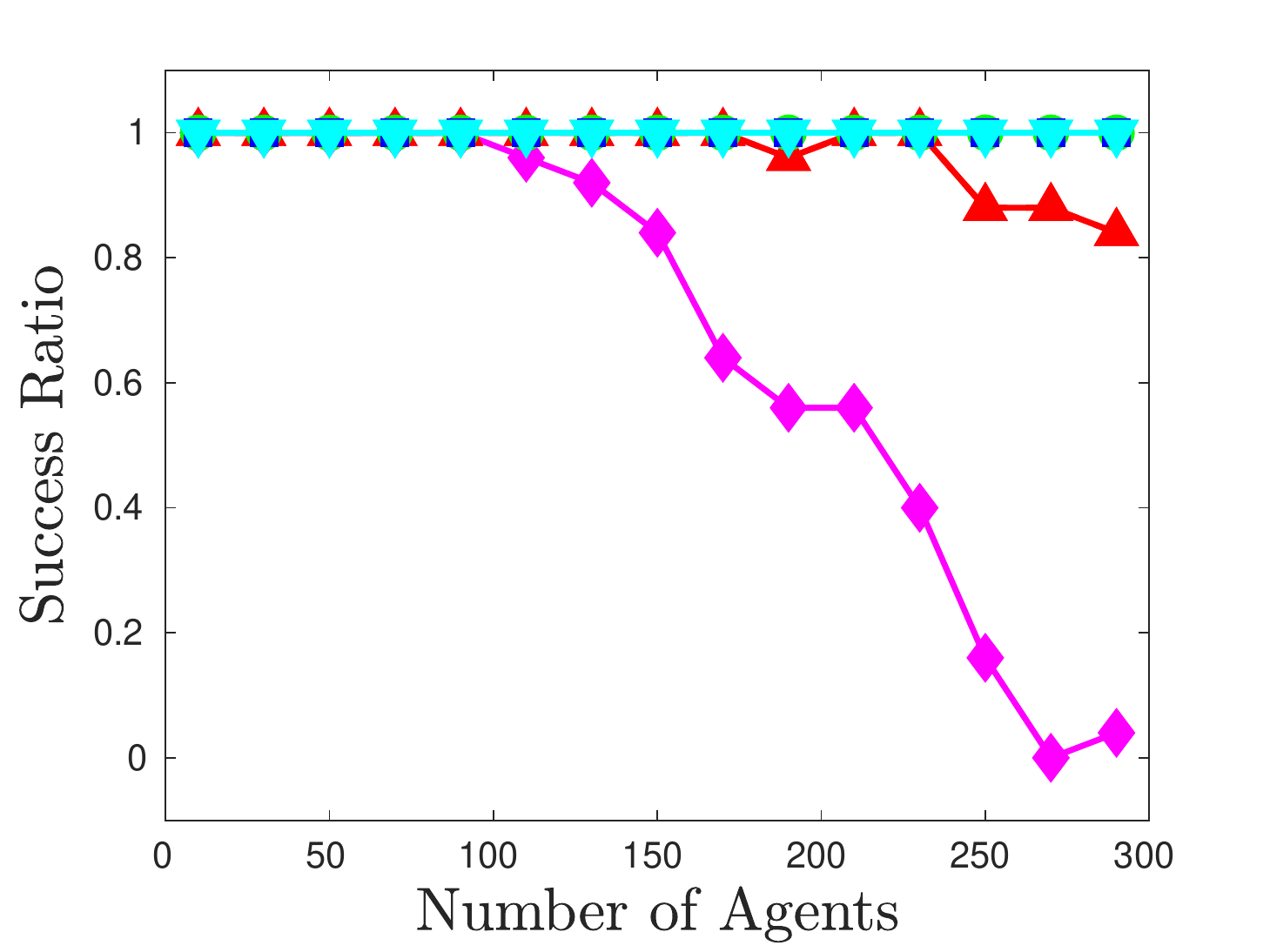}};
        \end{tikzpicture}
	\end{subfigure}\hfill
	\begin{subfigure}{\myMSFigureScale\textwidth}
        \begin{tikzpicture}
        \node(a){\includegraphics[width=\myLineScale\linewidth,trim={1mm 0 2mm 6mm},clip]{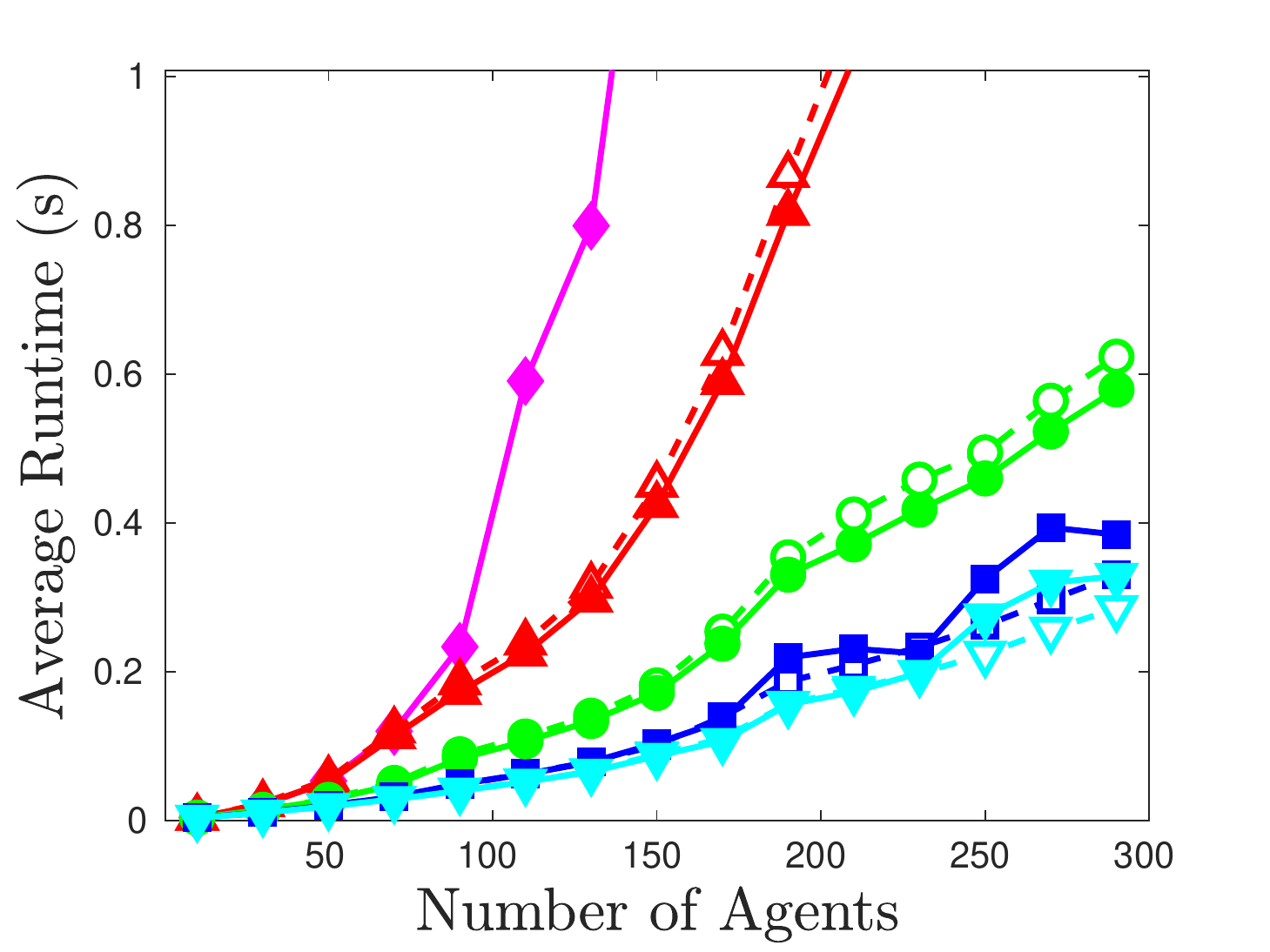}};
                \node at (a.north west)
        [
        anchor=center,
        xshift=12mm,
        yshift=-8mm
        ]
        {
            \includegraphics[width=0.3\linewidth]{figures/paris}
        };
             \node at (a.north west)
	    [
	    anchor=center,
	    xshift=12mm,
	    yshift=-16mm
	    ]
	    {
	    	\tiny{Paris 1 256}
	    };
        \end{tikzpicture}
	\end{subfigure}	
	\caption{Success ratio and average runtime over 25 random scenarios of CBS$^+$, ECBS$^+$, EECBS$^+$, and CBSB(-BP) with different numbers of agents in 6 different maps, given a time limit of 10 seconds and a memory limit of 16GB. Average runtime is plotted only if the corresponding solver found a solution in more than 75\% of the 25 random scenarios.}
	\label{cbsb:f:benchmark}
\end{figure*}


\begin{figure}[ht]
	\centering
	\def\breakdownScale{0.24}
	\begin{subfigure}{\breakdownScale\textwidth}
        \begin{tikzpicture}
        \node(a){\includegraphics[width=\myLineScale\linewidth, trim={2mm 0 2mm 6mm},clip]{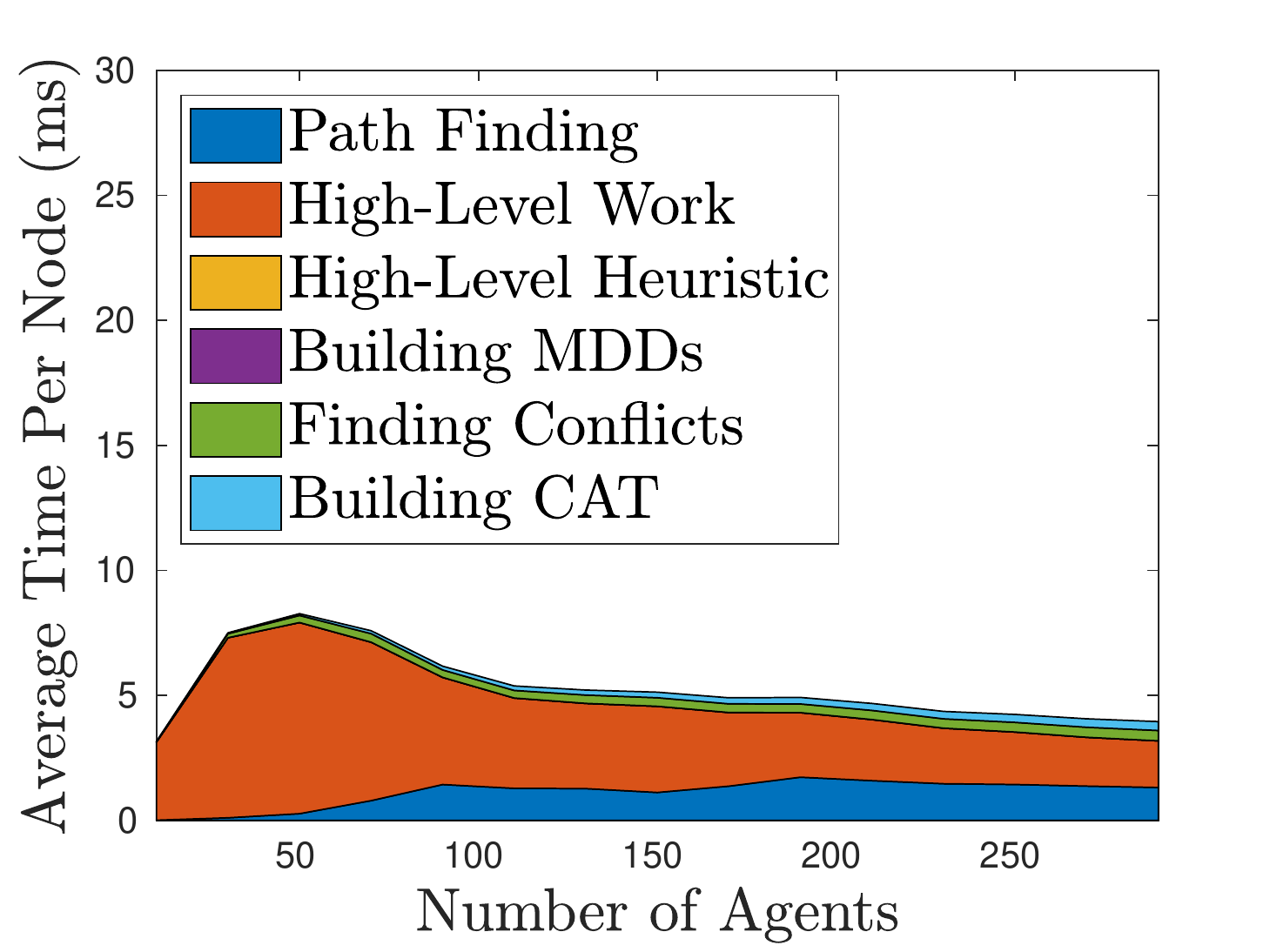}};
        \end{tikzpicture}
 		\caption{CBSB}
	\end{subfigure}
	\begin{subfigure}{\breakdownScale\textwidth}
        \begin{tikzpicture}
        \node(a){\includegraphics[width=\myLineScale\linewidth, trim={2mm 0 2mm 6mm},clip]{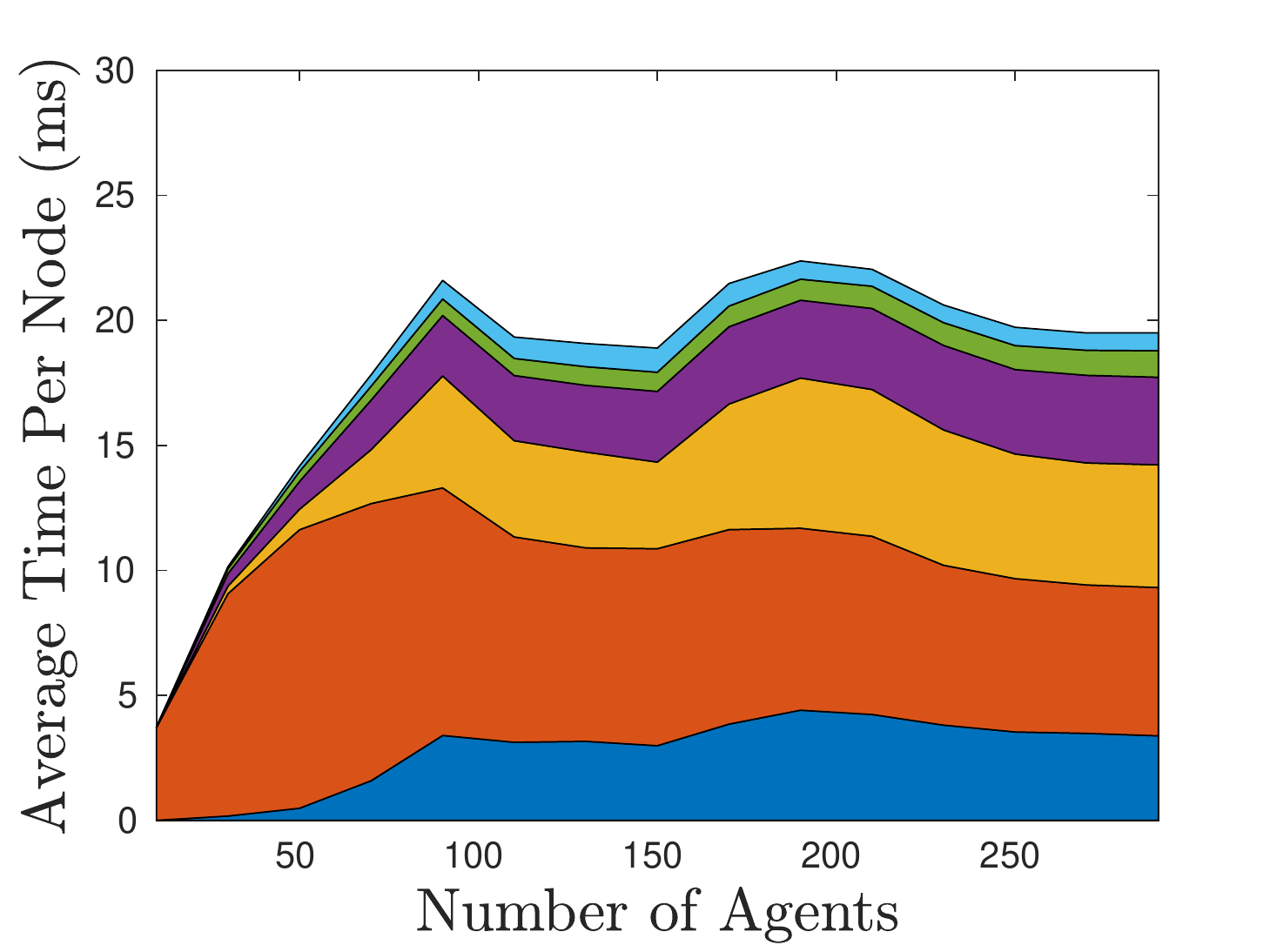}};
        \node at (a.north west)
        [
        anchor=center,
        xshift=12mm,
        yshift=-8mm
        ]
        {
            \includegraphics[width=0.23\linewidth]{figures/paris}
        };
         \node at (a.north west)
	    [
	    anchor=center,
	    xshift=12mm,
	    yshift=-15mm
	    ]
	    {
	    	\tiny{Paris 1 256}
	    };
        \end{tikzpicture}
 		\caption{EECBS$^+$}
	\end{subfigure}
	\caption{Average time per node over problem instances solved by both CBSB (left) and EECBS$^+$ (right) with the suboptimality factor of 1.2 in two different environments, namely, \texttt{Paris 1 256} 
	.}
	\label{cbsb:f:breakdown}
\end{figure}
\section{Numerical Experiments}

In this section, we evaluate CBSB on the standard MAPF benchmark suite~\cite{Stern2019} using six maps of different sizes with 15 different numbers of agents per map. 
For each map, we use 25 ``randomly" generated scenarios of the benchmark suite for a given number of agents, where each scenario contains different start and target locations of agents.
We compare five different CBS variants: CBS, ECBS, EECBS, CBSB, and CBSB-BP with 
suboptimality factors of 1.2 and 10 (i.e., $w=1.2$ and $w=10.0$) except for CBS, which finds the optimal solution. 
We use the modern implementation of CBS, ECBS, and EECBS found in \url{https://github.com/Jiaoyang-Li/EECBS},
and denote these algorithms CBS$^+$, ECBS$^+$, and EECBS$^+$, respectively. 
These algorithms are equipped with WDG heuristic~\cite{Li2019}, Prioritizing Conflicts~\cite{Boyarski2015a}, Bypassing Conflicts~\cite{Boyarski2015b}, and symmetry reasoning~\cite{Li2020}. 
All the experiments were given a time limit of 10 seconds, with a memory limit of 16GB. 
The algorithms were implemented in C++ and all the experiments were conducted on Ubuntu 18.04 LTS on an Intel Core i7-8750H with 16 GB of RAM.

Figure~\ref{cbsb:f:cost_comparison} plots the average solution cost over 25 random scenarios. The average solution cost is plotted only if more than 75\% of the 25 random scenarios are solved. 
In all cases, ECBS$^+$, EECBS$^+$, and CBSB find a near-optimal solution.
In Figure~\ref{cbsb:f:benchmark}, the success ratio and the average runtime over 25 random scenarios are plotted for different number of agents in six different maps. 
The average runtime is plotted only if more than 75\% of the 25 random scenarios are solved. 

As shown in Figure~\ref{cbsb:f:benchmark}, CBSB(-BP), and EECBS$^+$ outperform the other two algorithms, both in terms of the success rate and the average runtime to find a solution.
CBSB-BP always performs better than CBSB.
CBSB(-BP), and EECBS$^+$ show similar performances, and there is a no clear winner for all maps. 
CBSB(-BP) performed better in easier problem instances, where the success rates are high, whereas EECBS$^+$ performed better in harder problem instances.
EECBS$^+$ performs better than CBSB if finding a conflict minimal path is difficult due to many conflicts. 
CBSB(-BP) performs better than EECBS$^+$ if conflicts are relatively easy to be resolved at the expense of path length.   

Figure~\ref{cbsb:f:breakdown} shows the average runtime per generated CT node between CBSB and EECBS in two different environments, where CBSB performs better than EECBS.
\textsf{Path Finding} indicates the average time of the low-level search, and \textsf{High-Level Work} is the time it took at the high-level search.
\textsf{High-Level Heuristic} is the time it took for building a WDG heuristic graph and solving the edge-weighted minimum vertex cover.
\textsf{Building MDDs} indicates the time it took for building MDDs, and \textsf{Finding Conflicts} is the time to find conflicts among the agents in a CT node. 
\textsf{Building CAT} indicates the time it took for building the conflict avoidance table for each of the planning agents at a CT node in order to find any conflicting nodes at the low-level search.

In the \texttt{Paris 1 256} map (Figure~\ref{cbsb:f:breakdown}), CBSB spent less time than EECBS$^+$ both at the low-level and the high-level search.
At the high-level search, EECBS$^+$ spent more time to maintain the additional priority queues.
Also EECBS$^+$ has a noticeable overhead in computing the admissible high-level heuristic and building MDDs due to a large map size. 

We omit comparisons with the many MAPF solvers that have already been shown to perform worse than CBS~\cite{Wagner2011, Felner2012} or worse than ECBS~\cite{Barer2014} or worse than EECBS~\cite{Lam2020, Surynek2019}. 
We also omit comparisons with MAPF solvers that have no completeness and bounded-suboptimality guarantees~\cite{Li2022,Silver2021, Ma2019}. 
In summary, CBSB(-BP) shows state-of-the-art performance using a much simpler implementation, namely, 
without using the WDG heuristic~\cite{Li2019}, Prioritizing Conflicts~\cite{Boyarski2015a}, or Symmetry Reasoning~\cite{Li2020}, enhancements of CBS.

\section{Conclusion}

We have proposed a new bounded-cost MAPF algorithm, called CBSB, which uses a budgeted COA* (bCOA*) search at the low-level and a modified focal search at the high-level.
bCOA* can find a more informed heuristic plan, which is still bounded above, using a single priority queue.
We prove that, given a budget $B$, bCOA* returns a solution having a minimal number of conflicts, 
which is at most $B$-long, if one exists. 
If no solution exists that is shorter than $B,$ then bCOA* returns the shortest path.
Based on this property of bCOA*, a more informed suboptimality upper bound can be provided for the high-level search, allowing more exploration of minimal conflict plans.
We also prove that CBSB is complete, and returns a bounded-suboptimal solution, if one exists.
In the benchmark experiments, CBSB showed state-of-the-art performance despite its simple implementation.  
CBSB was able to compute a near-optimal solution within a fraction of second for hundreds of agents.


\bibliographystyle{plainnat}
\bibliography{bib/references.bib}

\end{document}